\documentclass[11pt,a4paper]{article}

\usepackage{hyperref}

\usepackage[fleqn]{amsmath}
\usepackage{amsfonts,amsthm,amssymb}

\usepackage{graphicx}

\usepackage{mathrsfs}
\usepackage{makeidx}
\usepackage{xypic}
\usepackage[nottoc]{tocbibind}
\usepackage{pstricks}
\usepackage{bbm}
\usepackage[arrow, matrix, curve]{xy}
\usepackage{epic}
\usepackage{eepic}
\usepackage{epsfig}

\numberwithin{equation}{section}

\usepackage{enumitem}

\xymatrixcolsep{0.5cm}
\xymatrixrowsep{0.5cm}

\newtheorem{theorem}{Theorem}[section]
\newtheorem{proposition}{Proposition}[section]
\newtheorem{lemma}{Lemma}[section]

\newtheorem{corollary}{Corollary}[section]

{\theoremstyle{remark}
\newtheorem{rem}{Remark}[section]}
{\theoremstyle{remark}
}
\newcommand{\tr}{\operatorname{tr}}

\newcommand{\End}{\operatorname{End}}

\newcommand{\bra}[1]{\langle\,#1\,|}
\newcommand{\ket}[1]{|\,#1\,\rangle}

\newcommand{\moy}[1]{\langle\,#1\,\rangle}

\def\la{\lambda}

\newcommand{\mathsc}[1]{{\normalfont\textsc{#1}}}

\begin{document}

\begin{flushright}
LPENSL-TH-11/14
\end{flushright}

\bigskip \vspace{15pt}

\begin{center}

\textbf{\Large Antiperiodic XXZ chains with arbitrary spins: Complete eigenstate construction by functional equations in separation of variables}


\vspace{45pt}

\begin{large}
{\bf G. Niccoli}\footnote{ENS Lyon; CNRS; Laboratoire de Physique, UMR 5672, Lyon France; giuliano.niccoli@ens-lyon.fr}
{\bf and V.~Terras}\footnote{Univ. Paris Sud; CNRS; LPTMS, UMR 8626, Orsay 91405 France; veronique.terras@lptms.u-psud.fr}
\end{large}


\vspace{45pt}

\today

\end{center}

\vspace{45pt}

\begin{abstract}
Generic inhomogeneous integrable XXZ chains with arbitrary spins are studied by means of the quantum separation of variables (SOV) method.
Within this framework, a complete description of the spectrum (eigenvalues and eigenstates) of the antiperiodic transfer matrix is derived in terms of discrete systems of equations involving  the inhomogeneity parameters of the model.
We show here that one can reformulate this discrete SOV characterization of the spectrum in terms of functional  $T$-$Q$ equations of Baxter's type, hence proving the completeness of the solutions to the associated systems of Bethe-type equations.
More precisely, we consider here two such reformulations.
The first one is given in terms of $Q$-solutions, in the form of trigonometric polynomials of a given degree $\mathsf{N}_s$, of a one-parameter family of $T$-$Q$ functional equations with an extra inhomogeneous term.
The second one is given in terms of $Q$-solutions, again in the form of trigonometric polynomials of degree $\mathsf{N}_s$ but with double period, of Baxter's usual (i.e. without extra term)  $T$-$Q$ functional equation.
In both cases, we prove the precise equivalence of the discrete SOV characterization of the transfer matrix spectrum with the characterization following from the consideration of the particular class of $Q$-solutions of the functional $T$-$Q$ equation: to each transfer matrix eigenvalue corresponds exactly one such $Q$-solution and {\em vice versa}, and this $Q$-solution can be used to construct the corresponding eigenstate.
\end{abstract}

\newpage

\section{Introduction}

In this paper we consider quantum one-dimensional integrable models defined from height weight representations of the trigonometric 6-vertex Yang-Baxter algebra with arbitrary positive integer dimension in any site of the lattice.
This class of models contains in particular the integrable XXZ chains with arbitrary spins when a special (homogeneous) limit is taken on the parameters of these representations.

In the case of periodic boundary conditions (with possibly an additional diagonal twist), the eigenvalues and eigenvectors  of the one-parameter family of transfer matrices of the model have been obtained in the framework of the algebraic Bethe ansatz (ABA) method \cite{SklF78,SklTF79,FadT79,KirR87}.
In this framework it has also been possible to derive explicit representations for interesting physical quantities such as the correlation functions \cite{KitMT00,Kit01,KitMST02a,GohKS04,KitMST05a,KitMST05b,CasM07,DegM11,Deg12}. In the spin-1/2 case, this notably led, on the one hand, to very accurate numerical results concerning the dynamical structure factors \cite{CauM05,CauHM05,PerSCHMWA06,PerSCHMWA07}, and, on the other hand, to analytical results concerning the large-distance asymptotic behavior of the two-point \cite{KitKMST09a,KitKMST09b,KozT11,KitKMST11b,KitKMST12,DugGK13} and even the multi-point \cite{KitKMT14} correlation functions.
At the root of these successes were the ABA characterization of the eigenstates, the explicit computation of their scalar product \cite{Sla89} and the solution of the so-called quantum inverse problem  \cite{KitMT99,MaiT00}, which were used to get simple determinant representations \cite{KitMT99} for the form factors (i.e. the matrix elements of local operators in the basis of the transfer matrix eigenstates) of the model. Note that some of these results have been extended \cite{KitKMNST07,KitKMNST08} to open integrable quantum models \cite{Che84,Skl88} for the special boundary conditions for which ABA still applies \cite{Skl88}.

In the case of antiperiodic or twisted-antiperiodic boundary conditions, however, ABA can no longer be used to construct the eigenstates of the corresponding transfer matrices. In fact, for a long time, only their eigenvalues had been studied by means of various functional methods, such as Baxter's $T$-$Q$ functional relation \cite{BatBOY95,YunB95} based on the ``pair-propagation through a vertex" property  \cite{Bax72,Bax82L} combined with the fusion of transfer matrices \cite{KulS82,Res83,KirR86a,KirR87}, or, in the root of unity case, by functional relations issued from the fusion hierarchy and truncation identities \cite{NieWF09}.
Let us also mention, in the spin 1/2 case, the derivation of alternative functional relations in \cite{Gal08}, and the use of a functional version of Sklyanin's separation of variables method \cite{Skl85,Skl90,Skl92,Skl95} in \cite{NieWF09}.
For these antiperiodic integrable quantum models, the absence of an eigenstate construction has determined a lack of knowledge concerning the computation of further physical quantities such as the correlation functions, even more marked when compared to the richness of results obtained instead in the periodic case. 

Yet, Sklyanin's quantum separation of variables method (SOV) \cite{Skl85,Skl90,Skl92,Skl95} is a very powerful approach to tackle these models, complementary to ABA in the sense that it enables one to consider various boundary conditions (including (twisted-)antiperiodic ones) for which ABA cannot directly be applied. It provides a full characterization of the corresponding transfer matrix eigenvalues and eigenstates, the completeness being here a direct consequence of the SOV approach. This is a considerable advantage with respect to the ABA for which completeness is in general a complicated issue (see for instance \cite{Dor93,TarV95}).
Another advantage of this method is that determinant representations for the eigenstates scalar products are obtained by construction, whereas Slavnov's scalar product formula \cite{Sla89}, one of the fundamental ingredient of the ABA approach to correlation functions, is a highly non-trivial identity which has to be ``re-invented" as soon as one modifies the model (see for instance \cite{BelPRS12b,LevT13a}).
Indeed, it has been proven for a large variety of integrable quantum models \cite{GroMN12,Nic13a,GroMN14}, including the spin-1/2 XXZ chain \cite{Nic13} and the higher spin integrable generalizations of the XXX chain \cite{Nic13b}, that, in the SOV framework, the scalar products associated to the class of separate states (containing in particular all the transfer matrix eigenstates) admit a universal simple form written as a determinant of matrices described by simple linear combinations of Vandermonde matrices. This result, together with the aforementioned solution of the quantum inverse problem, allows one to get simple determinant representations for the matrix elements of local operators. The hope is that these representations could be the starting point for the effective computation of the correlation functions, similarly to what happens in the periodic case where the form factor representations \cite{KitMT99} have given access to the correlation functions. 
The problem at this stage is that, unlike in the periodic case \cite{IzeKMT99,KitKMST09c,KitKMST11a}, it is not completely clear yet how to evaluate the behavior of these determinants at the thermodynamic limit.
It is even not so obvious to consider their homogeneous limit.
This comes mainly from the fact that the SOV characterization of the spectrum relies on a discrete system of quadratic equations involving directly the inhomogeneity parameters of the model, a description which is neither convenient for the consideration of the thermodynamic limit nor for the consideration of the homogeneous limit.
Hence, so as to have a chance to get some quantitative evaluation of physical quantities such as the correlation functions starting from the SOV solution of the model, it seems necessary to be able to reformulate this spectrum characterization in a more convenient way, for instance in terms of some solutions of Bethe-type equations, as in the periodic case.

In fact, the SOV discrete system of equations for the spectrum of the model is {\em a priori} not completely disconnected from a system of Bethe equations. Indeed, it appears as a discrete version, evaluated at the inhomogeneity parameters of the model, of Baxter's famous $T$-$Q$ functional equation relating two ({\em a priori} unknown) classes of entire functions of the spectral parameter $\lambda$: the transfer matrix eigenvalues $\mathsf{t}(\lambda)$ and the eigenvalues $Q(\lambda)$ of the so-called $Q$-operator \cite{Bax72,Bax82L}. This functional equation is a second-order finite-difference homogeneous equation which takes the following form
\begin{equation}\label{hom-TQ}
   \mathsf{t}(\lambda)\, Q(\lambda)=\mathsf{a}(\lambda)\, Q(\lambda+\epsilon_q)+\mathsf{d}(\lambda)\, Q(\lambda-\epsilon_q),
\end{equation}
where $\mathsf{a}(\lambda)$, $\mathsf{d}(\lambda)$ are some functions of the spectral parameter depending on the model we consider (i.e. related to the corresponding representation of the Yang-Baxter algebra), and $\epsilon_q$ is some finite shift which is a parameter of the model.
If one knows the functional form of the eigenvalues $Q(\lambda)$ of the $Q$-operator, then the functional equation \eqref{hom-TQ} is in its turn equivalent to a system of Bethe equations for the roots of the function $Q(\lambda)$, translating the fact that the eigenvalues of the transfer matrix should be entire functions of the spectral parameter.

However, to our knowledge, the equivalence of the SOV discrete characterization of the spectrum and eigenstates with a particular class of solutions of an homogeneous equation of the form \eqref{hom-TQ} has never been clearly demonstrated, even in the simple case of the antiperiodic XXZ spin-1/2 chain. In fact, it should be noticed that, although all functions of the model are $2i\pi$-periodic, and contrary to what happens in the periodic case \cite{Bax82L}, the functional equation \eqref{hom-TQ} for the {\em antiperiodic} transfer matrix eigenvalues does not admit any solution of the form
\begin{equation}\label{Q-2pi}
  Q(\lambda)=\prod_{j=1}^\mathsf{M}\sinh(\lambda-\lambda_j),
  \qquad \mathsf{M}\in\mathbb{N},
  \qquad \lambda_1,\ldots,\lambda_{\mathsf{M}}\in\mathbb{C}.
\end{equation}
This is due to the fact that the eigenvalue $\bar{\mathsf{t}}(\lambda)$ of the antiperiodic transfer matrix, contrary to its periodic counterpart $\mathsf{t}(\lambda)$, does not satisfy the same $i\pi$-quasi-periodicity property as the two functions $\mathsf{a}(\lambda)$ and $\mathsf{d}(\lambda)$.
The construction of the $Q$-operator performed in \cite{BatBOY95} leads instead to solutions of the form
\begin{equation}\label{Q-4pi}
  Q(\lambda)=\prod_{a=1}^{\mathsf{N}}\sinh \left( \frac{\lambda -\lambda_j}{2}\right) ,
  \qquad \lambda_1,\ldots,\lambda_{\mathsf{N}}\in\mathbb{C},
\end{equation}
which are only $4i\pi$-periodic ($\mathsf{N}$ denotes here the size of the system).
Solutions of the form \eqref{Q-4pi}, contrary to solutions of the form \eqref{Q-2pi}, are indeed clearly compatible with the $i\pi$-quasi-periodicity of \eqref{hom-TQ} in the antiperiodic case. What is {\em a priori}  less clear is whether these solutions provide a {\em complete} description of the spectrum, especially considering the fact that the $Q$-operator construction of \cite{BatBOY95} relies on a conjecture.
It has even been suggested recently, in the context of another approach, the so-called {\em off-diagonal Bethe ansatz} \cite{CaoYSW13a}, that it could be easier and more natural to keep considering, as in the periodic case, functions of the form \eqref{Q-2pi}, as solutions of some generalized version of the functional equation \eqref{hom-TQ} of the form
\begin{equation}\label{inhom-TQ}
   \mathsf{t}(\lambda)\, Q(\lambda)=\tilde{\mathsf{a}}(\lambda)\, Q(\lambda+\epsilon_q)+\tilde{\mathsf{d}}(\lambda)\, Q(\lambda-\epsilon_q)+F(\lambda),
\end{equation}
i.e. with some {\em inhomogeneous} term $F(\lambda)$ and some modified functions $\tilde{\mathsf{a}}(\lambda)$ and $\tilde{\mathsf{d}}(\lambda)$. But, even in this case, the completeness of such solutions has not been, to our opinion, correctly demonstrated.

The purpose of the present paper is to fill these gaps.
We first obtain, within the SOV framework, the complete set of eigenstates of the twisted antiperiodic transfer matrix for the inhomogeneous model associated to any arbitrary spin representations of the 6-vertex trigonometric Yang-Baxter algebra, hence generalizing the construction introduced by one of the authors for the antiperiodic XXZ spin-1/2 chain \cite{Nic13} and the antiperiodic XXX chains with arbitrary spin \cite{Nic13b}. We notably prove that the corresponding spectrum is non-degenerate. 
Second, we reformulate the SOV complete characterization of the spectrum and eigenvectors in terms of solutions of the form \eqref{Q-2pi} of a one-parameter family of {\em inhomogeneous} functional equations of the form \eqref{inhom-TQ}, corresponding to higher spin generalizations of an equation proposed in the context of the off-diagonal Bethe ansatz in \cite{ZhaLCYSW14}. The arguments that we use here are very simple, relying essentially on the equivalence of this functional equation with its discrete version at a sufficient number of well-chosen independent points. This proves the completeness of the characterization of the spectrum in terms of these solutions and of the corresponding Bethe-type equations, although these are probably not the most interesting for the resolution of the present model.
Finally, we reformulate the SOV complete characterization of the spectrum and eigenvectors in terms of the solutions of the form \eqref{Q-4pi} (or their higher spin generalization \cite{YunB95}) of Baxter's homogeneous functional equation. The arguments that we use to prove this correspondence are hardly more complicated (and even simpler in some aspects) than in the inhomogeneous case, also relying on the consideration of the discretization of \eqref{hom-TQ} at an appropriate set of values. We insist on the fact that they do not involve the algebraic construction of the $Q$-operator (although they finally provide an implicit definition of the latter in every finite-dimensional representation of the Yang-Baxter algebra, as a diagonal matrix in the eigenbasis of the antiperiodic transfer matrix). This proves the completeness of the description of the spectrum and eigenvectors in terms of the solutions of the Bethe equations following from \eqref{hom-TQ}-\eqref{Q-4pi} (and their higher spin generalizations). We expect that this last reformulation will be useful for the consideration of the homogeneous/thermodynamic limit of the form factor determinant representations obtained in the SOV framework, and hence will effectively constitute the first step of our future analysis of the quantum dynamics of these integrable quantum models.

\section{The integrable XXZ chain with arbitrary spins}
\label{sec-XXZ-def}

The integrable XXZ chain with arbitrary spins is defined in the framework of the Quantum Inverse Scattering Method (QISM) \cite{SklTF79,FadT79,KulR81,KulS82,KirR87} by means of local quantum Lax operators of the form:
\begin{equation}\label{q-Lax}
\mathsf{L}_{0n}^{(1/2,s_{n})}(\lambda )
= 
\begin{pmatrix}
\sinh ( \lambda +\eta S_{n}^{z} ) & S_{n}^{-}\, \sinh \eta \\ 
S_{n}^{+}\, \sinh \eta & \sinh( \lambda -\eta S_{n}^{z})
\end{pmatrix}_{\![0]}\in \End(V_{0}^{(1/2)}\otimes V_{n}^{(s_{n})}),
\end{equation}
in which $V_0^{(1/2)}\simeq \mathbb{C}^2$ corresponds to the 2-dimensional auxiliary space whereas $V_n^{(s_n)}\simeq\mathbb{C}^{2 s_n+1}$ corresponds to the local $(2s_n+1)$-dimensional quantum space at site $n$, where $2 s_n\in\mathbb{N}$.
The local spin operators $S_n^{\pm}, S_n^z\in\End(V_n^{(s_n)})$ realize a spin-$s_n$ representation of the quantum algebra $U_q(\mathfrak{sl}_2)$ on $V_n^{(s_n)}$,
\begin{equation}\label{Uq(sl2)}
 [ S^{z},S^{\pm }]=\pm S^{\pm },
 \qquad
 [S^{+},S^{-}]
 =\frac{q^{2S^{z}}-q^{-2 S^z}}{q-q^{-1}}=\frac{\sinh(2\eta S^{z})}{\sinh\eta},
 \qquad q=e^\eta.
\end{equation}
In this paper, $\eta$ is a generic complex parameter, non commensurate to $i\pi$ (i.e. $q$ is not a root of unity).

The quantum Lax operator \eqref{q-Lax} satisfies the quadratic relation
\begin{equation}
R_{0 0'}(\lambda -\mu )\,\mathsf{L}_{0n}^{(1/2,s_{n})}(\lambda )\, \mathsf{L}_{0' n}^{(1/2,s_{n})}(\mu )
=\mathsf{L}_{0' n}^{(1/2,s_{n})}(\mu )\, \mathsf{L}_{0n}^{(1/2,s_{n})}(\lambda )\, R_{0 0'}(\lambda -\mu ),  \label{RLL}
\end{equation}
on $V_{0}^{(1/2)}\otimes V_{0'}^{(1/2)}\otimes V_{n}^{(s_{n})}$
in terms of the trigonometric 6-vertex R-matrix
\begin{equation}\label{mat-R}
R(\lambda )=
\begin{pmatrix}
\sinh \left( \lambda +\eta \right) & 0 & 0 & 0 \\ 
0 & \sinh \lambda & \sinh \eta & 0 \\ 
0 & \sinh \eta & \sinh \lambda & 0 \\ 
0 & 0 & 0 & \sinh (\lambda +\eta )
\end{pmatrix} 
=\mathsf{L}^{(1/2,1/2)}\big(\lambda+\frac\eta 2 \big).
\end{equation}
As usual, here and in the following, the indices label the spaces of the tensor product on which the corresponding operators act in a non-trivial way.
The monodromy matrix of an inhomogeneous chain of size $\mathsf{N}$, defined on $V_0^{(1/2)}\otimes\mathcal{H}_{\mathsf{N}}$ as
\begin{equation}\label{mon}
\mathsf{M}_{0}^{(1/2)}(\lambda )
\equiv 
\mathsf{L}_{0\mathsf{N}}^{(1/2,s_{\mathsf{N}})}(\lambda -\xi_{\mathsf{N}})\cdots \mathsf{L}_{01}^{(1/2,s_{1})}(\lambda -\xi_{1})
=\begin{pmatrix}
\mathsf{A}(\lambda ) & \mathsf{B}(\lambda ) \\ 
\mathsf{C}(\lambda ) & \mathsf{D}(\lambda )
\end{pmatrix}_{\! [0]},
\end{equation}
satisfies on $V_{0}^{(1/2)}\otimes V_{0'}^{(1/2)}\otimes \mathcal{H}_{\mathsf{N}}$ a quadratic equation similar to \eqref{RLL}:
\begin{equation}\label{RTT}
R_{0 0'}(\lambda -\mu )\,\mathsf{M}_{0}^{(1/2)}(\lambda )\,\mathsf{M}_{0'}^{(1/2)}(\mu )
=\mathsf{M}_{0'}^{(1/2)}(\mu )\,\mathsf{M}_{0}^{(1/2)}(\lambda )\,R_{0 0'}(\lambda -\mu ).
\end{equation}
Here $\xi_n$ is the inhomogeneity parameter at site $n$, $n\in\{1,\ldots,\mathsf{N} \}$, and $\mathcal{H}_{\mathsf{N}}\equiv\otimes _{n=1}^{\mathsf{N}}V_{n}^{(s_{n})}$ is the quantum space of states of the chain.
The operators $\mathsf{A}(\lambda )$, $\mathsf{B}(\lambda )$, $\mathsf{C}(\lambda )$, $\mathsf{D}(\lambda ) \in \End(\mathcal{H}_{\mathsf{N}})$, which satisfy the commutation relations given by \eqref{RTT}, hence provide a representation of the Yang-Baxter algebra associated to \eqref{mat-R}.
The latter admits a central element, the so-called quantum determinant, which can be expressed as
\begin{align}
  \mathrm{det}_q \mathsf{M}^{(1/2)}(\lambda)=a(\lambda)\, d(\lambda-\eta)
  &=\mathsf{A}(\lambda)\,\mathsf{D}(\lambda-\eta)-\mathsf{B}(\lambda)\,\mathsf{C}(\lambda-\eta)
       \nonumber\\
  &=\mathsf{D}(\lambda)\,\mathsf{A}(\lambda-\eta)-\mathsf{C}(\lambda)\,\mathsf{B}(\lambda-\eta)
      ,\label{q-det}
\end{align}
where
\begin{equation}\label{a-d}
   a(\lambda)=\prod_{n=1}^\mathsf{N}\sinh(\lambda-\xi_n+s_n\eta),
   \qquad
   d(\lambda)=\prod_{n=1}^\mathsf{N}\sinh(\lambda-\xi_n-s_n\eta).
\end{equation}

Let $\mathsf{K}$ be a numerical $2\times 2$ matrix such that $[R(u),\mathsf{K}\otimes\mathsf{K}]=0$. Then the $\mathsf{K}$-twisted transfer matrix
\begin{equation}
  \mathsf{T}_\mathsf{K}(\lambda)=\tr_0\big[\mathsf{K}_0\,\mathsf{M}_0^{(1/2)}(\lambda)\big]
\end{equation}
generates a one-parameter family of commuting operators:
\begin{equation}
    \big[\mathsf{T}_\mathsf{K}(\lambda),\mathsf{T}_\mathsf{K}(\mu)\big]=0,
    \qquad\forall\lambda,\mu\in\mathbb{C}.
\end{equation}
The family of $\kappa$-twisted periodic transfer matrix $\mathsf{T}^{(\kappa)}(\lambda)$ corresponding to $\mathsf{K}=\mathrm{diag}(\kappa,\kappa^{-1})$ for some given $\kappa\in\mathbb{C}\setminus\{0\}$ can be diagonalized by algebraic Bethe ansatz, leading to a parametrization of the corresponding eigenvalues in terms of the solutions of a system of $\kappa$-twisted Bethe equations.
Here, we are instead interested in the diagonalization of the family of $\kappa$-twisted {\em antiperiodic} transfer matrices
\begin{equation}\label{anti-T}
 \bar{\mathsf{T}}^{(\kappa)}(\lambda)=\kappa^{-1}\,\mathsf{B}(\lambda)+\kappa\,\mathsf{C}(\lambda),
\end{equation}
corresponding to $\mathsf{K}=\mathrm{diag}(\kappa,\kappa^{-1})\cdot\sigma^x$, by means of separation of variables.

Without further study, it is already easy to point out certain conditions under which this transfer matrix is normal (and hence diagonalizable):
\begin{proposition}\label{prop-normal}
\textsf{I)} Let $\eta \in i\mathbb{R}$, $(\xi_1,\ldots,\xi_{\mathsf{N}}) \in \mathbb{R}^{\mathsf{N}}$ and $\kappa^*=\kappa^{-1}$.
Then $\left[\mathsf{\bar{T}}^{(\kappa)}(\lambda )\right]^\dagger=-\mathsf{\bar{T}}^{(\kappa)}(\lambda^* )$, and therefore $\mathsf{\bar{T}}^{(\kappa)}(\lambda )$ is normal for any $\lambda\in\mathbb{C}$.

\textsf{II)} Let $\eta \in \mathbb{R}$, $(\xi_1,\ldots,\xi_{\mathsf{N}}) \in ( i\mathbb{R})^{\mathsf{N}}$ and $\kappa^*=\kappa^{-1}$.
Then $\left[\mathsf{\bar{T}}^{(\kappa)}(\lambda )\right]^\dagger=(-1)^{\mathsf{N}-1}\mathsf{\bar{T}}^{(\kappa)}(-\lambda^* )$, and therefore $\mathsf{\bar{T}}^{(\kappa)}(\lambda )$ is normal for any $\lambda\in\mathbb{C}$.
\end{proposition}

\section{Diagonalization of antiperiodic transfer matrices in the SOV framework}
\label{sec-SOV}

The $\bar{\mathsf{T}}$-spectral problem of the ($\kappa$-twisted) antiperiodic model can be solved by means of the SOV approach \cite{Skl85,Skl90,Skl92,Skl95}.
This approach is based on the construction of a basis of the space of states for which the commutative family of operators $\mathsf{D}(\lambda )$ (or $\mathsf{A}(\lambda )$) is diagonal and with
simple spectrum.
The construction of such a basis has been explicitly performed, in the XXZ spin-1/2 case, in \cite{Nic13}, and in the integrable XXX chain with arbitrary spin in \cite{Nic13b}. The construction here is very similar, therefore we only recall the main results.

In the following, we shall denote by a subscript $\mathcal{L}$ (respectively $\mathcal{R}$) the left (respectively the right) representation spaces. We shall denote by $\ket{k,n}$ (respectively by $\bra{k,n}$), $k\in\{0,1,\ldots,2s_{n}\}$, the vector of the $S_{n}^{z}$-eigenbasis of the local
space $V_{n}^{(s_{n},\mathcal{R})}$ (respectively of the local
space $V_{n}^{(s_{n},\mathcal{L})}$),  such that
\begin{alignat}{2}
  &S_{n}^{z} \, \ket{k,n} =(s_{n}-k)\, \ket{k,n},
  \qquad
  & &\bra{k,n} \, S_{n}^{z} =(s_{n}-k)\, \bra{k,n},\\
  &S_{n}^{-} \, \ket{k,n} =x_n(k+1)\, \ket{k+1,n},
  \qquad
  & & \bra{k,n}\, S_{n}^{-}=x_n(k)\, \bra{k-1,n},\\
  &S_{n}^{+} \, \ket{k,n} =x_n(k)\, \ket{k-1,n},
  \qquad
  & &\bra{k,n}\, S_{n}^{+}=x_n(k+1)\, \bra{k+1,n},
\end{alignat}
with $x_n(k)=\sqrt{[k]_q\, [2s_{n}-k+1]_q}$,
where we have used the standard  $q$-integer notation:
\begin{equation}\label{q-int}
   [ j]_q\equiv\frac{q^j-q^{-j}}{q-q^{-1}}=\frac{\sinh(\eta j)}{\sinh\eta}.
\end{equation}
We have the natural scalar product:
\begin{equation}
(\ket{k,n},\ket{k',n})\equiv\moy{k,n\,|\, k',n} = \delta _{k,k'},
\quad
\forall k,k'\in \{0,\ldots,2s_{n}\}.
\end{equation}
Setting
\begin{equation}
\bra{ \mathbf{0} }\equiv \frac{1}{\mathsc{n}}\otimes _{n=1}^{\mathsf{N}}\bra{0,n},
\qquad
\ket{\mathbf{0} } \equiv \frac{1}{\mathsc{n}}\otimes _{n=1}^{\mathsf{N}}\ket{0,n} ,
\qquad \text{with}\quad
\mathsc{n}\equiv \prod_{i<j}\sinh(\xi_j-\xi_i)^{1/2},
\end{equation}
to be respectively the left and right {\em references states},
we define, for each $\mathsf{N}$-tuple $\mathbf{h}\equiv(h_1,\ldots,h_\mathsf{N})$ with $h_n\in\{0,1,\ldots,2s_n\}$, the following states in the {\em left} (covectors) and {\em right} (vectors) linear spaces:
\begin{alignat}{2}
   &\bra{\mathbf{h}}
   \equiv  \bra{\mathbf{0}}\,
   \prod_{n=1}^{N}\prod_{k_n=0}^{h_n-1} \frac{C(\xi_n^{(k_n)})}{d(\xi_n^{(k_n+1)})}
   &\quad
   &\in\ \mathcal{H}^\mathcal{L}_{\mathsf{N}}\equiv \otimes _{n=1}^{\mathsf{N}} V_{n}^{(s_{n},\mathcal{L})}, 
    \label{leigen-h}
    \\
    &\ket{\mathbf{h}} \equiv 
    \prod_{n=1}^{\mathsf{N}}\prod_{k_{n}=0}^{h_{n}-1}
    \left( -\frac{B(\xi_n^{(k_n)})}{a(\xi_n^{(k_n)})}\right) \ket{\mathbf{0}}
    &\quad
    &\in\  \mathcal{H}^\mathcal{R}_{\mathsf{N}}\equiv \otimes_{n=1}^{\mathsf{N}}V_{n}^{(s_{n},\mathcal{L})}, 
    \label{reigen-h}
\end{alignat}
where we have used the notation
\begin{equation}\label{xi-shift}
   \xi_n^{(k_n)}\equiv\xi_n+(s_n-k_n)\eta  \qquad \text{for}\quad
   k_n\in\{0,1,\ldots, 2s_n\}.
\end{equation}
Then the family of covectors \eqref{leigen-h} (respectively of vectors \eqref{reigen-h}) provides a $\mathsf{D}$-eigenbasis in $\mathcal{H}^\mathcal{L}_{\mathsf{N}}$ (respectively in $\mathcal{H}^\mathcal{R}_{\mathsf{N}}$). More precisely, we have the following theorem, which can be proven similarly as in \cite{Nic13,Nic13b}:

\begin{theorem}\label{th-SOV-basis}
Let the inhomogeneities $(\xi_1,\ldots,\xi_{\mathsf{N}})\in \mathbb{C}^{\mathsf{N}}$ satisfy the condition
\begin{equation}\label{cond-inh}
  \Xi_i\cap\Xi_j=\emptyset
   \qquad \text{for}\quad
    i\neq j,
\end{equation}
where $\Xi_i=\{\xi_i+k\eta+k'\pi\mid k=-s_i,-s_i+1,\ldots,s_i; k'\in\mathbb{Z}\}$.
Then $\mathsf{D}(\lambda)$ is diagonalizable with simple spectrum and we have:
\smallskip

\textsf{I)} \underline{Left $\mathsf{D}(\lambda )$ SOV-representations:}
The states $\bra{\mathbf{h}}$ \eqref{leigen-h} define a $\mathsf{D}$-eigenbasis of $\mathcal{H}^{\mathcal{L}}_{\mathsf{N}}$ such that
\begin{align}
   &\bra{\mathbf{h}}\, \mathsf{D}(\lambda )=d_{\mathbf{h}}(\lambda)\,\bra{\mathbf{h}}, 
      \label{D-left}\\
   &\bra{\mathbf{h}}\, \mathsf{C}(\lambda )
      =\sum_{a=1}^{\mathsf{N}}\prod_{b\neq a}
        \frac{\sinh \big( \lambda -\xi_b^{(h_b)}\big) }{\sinh \big( \xi_a^{(h_a)}-\xi_b^{(h_b)}\big) }\,
        (1-\delta_{h_a,2s_a})\, d(\xi_a^{(h_a+1)})\, \bra{\mathrm{T}_a^+\mathbf{h} },
\label{C-left} \\
   &\bra{\mathbf{h}}\, \mathsf{B}(\lambda )
      =-\sum_{a=1}^{\mathsf{N}}\prod_{b\neq a}
         \frac{\sinh \big( \lambda -\xi_b^{(h_b)}\big) }{\sinh \big( \xi_a^{(h_a)}-\xi_b^{(h_b)}\big) }\,
         (1-\delta_{h_a,0})\, a(\xi_a^{(h_a-1)})\, \bra{\mathrm{T}_a^{-}\mathbf{h} }.
\label{B-left}
\end{align}
\smallskip

\textsf{II)} \underline{Right $\mathsf{D}(\lambda )$ SOV-representations:} 
The states $\ket{\mathbf{h}}$ \eqref{reigen-h} define a $\mathsf{D}$-eigenbasis of $\mathcal{H}^{\mathcal{R}}_{\mathsf{N}}$ such that
\begin{align}
   &\mathsf{D}(\lambda )\,\ket{\mathbf{h}} =d_{\mathbf{h}}(\lambda)\,\ket{\mathbf{h}}, 
      \label{D-right}\\
   &\mathsf{C}(\lambda )\,\ket{\mathbf{h}} 
      =\sum_{a=1}^{\mathsf{N}}\prod_{b\neq a}
        \frac{\sinh \big( \lambda -\xi_b^{(h_b)}\big) }{\sinh \big( \xi_a^{(h_a)}-\xi_b^{(h_b)}\big) }\,
        d(\xi_a^{(h_a)})\, \ket{\mathrm{T}_a^-\mathbf{h} },
\label{C-right} \\
   & \mathsf{B}(\lambda )\,\ket{\mathbf{h}} 
      =-\sum_{a=1}^{\mathsf{N}}\prod_{b\neq a}
         \frac{\sinh \big( \lambda -\xi_b^{(h_b)}\big) }{\sinh \big( \xi_a^{(h_a)}-\xi_b^{(h_b)}\big) }\,
         a(\xi_a^{(h_a)})\, \ket{\mathrm{T}_a^+\mathbf{h} }.
\label{B-right}
\end{align}
In the above expressions, we have used the notations:
\begin{align}
 &d_{\mathbf{h}}(\lambda )\equiv \prod_{n=1}^{\mathsf{N}}\sinh (\lambda-\xi_n^{(h_n)}), 
 \label{eigen-D}\\
  &\mathrm{T}_a^\pm(h_1,\ldots,h_\mathsf{N})=(h_1,\ldots,h_a\pm 1,\ldots,h_\mathsf{N}).
\end{align}
The action of $\mathsf{A}(\lambda )$ on $\bra{\mathbf{h}}$ or on $\ket{\mathbf{h}}$ is uniquely determined by the quantum determinant relation \eqref{q-det}.

Moreover, the action of the covector $\bra{\mathbf{h}}$ on the vector $\ket{\mathbf{k}}$ reads:
\begin{equation} \label{sc-hk}
   \moy{\mathbf{h}\, |\, \mathbf{k} } 
   =\prod_{j=1}^{\mathsf{N}}\delta _{h_j,k_j}
      \prod_{1\leq i<j\leq \mathsf{N}}\frac{1}{\sinh \big(\xi_j^{(h_j)}-\xi_i^{(h_i)}\big) },
\end{equation}
which means that we have  the following decomposition
of the identity:
\begin{equation}\label{Id}
  \mathbb{I}\equiv \sum_{\mathbf{h}}
  \prod_{1\leq i<j\leq \mathsf{N}}\sinh \big(\xi_j^{(h_j)}-\xi_i^{(h_i)}\big)\,
  \ket{\mathbf{h}}\bra{\mathbf{h}},
\end{equation}
in which the summation runs over all $\mathsf{N}$-tuples $\mathbf{h}=(h_1,\ldots,h_\mathsf{N})$ such that $h_i\in\{0,1,\ldots,2s_i\}$.

\end{theorem}

Let us comment here that expressions of the type \eqref{D-left}-\eqref{B-left} and \eqref{D-right}-\eqref{B-right} for the generators of the trigonometric 6-vertex Yang-Baxter algebra can also be derived from the original representations by constructing the factorizing $F$-matrices and implementing the prescribed change of basis \cite{MaiS00}. Such $F$-matrices were introduced in \cite{MaiS00} as explicit representations of particular Drinfeld twist for quasi-triangular quasi-Hopf algebras \cite{Dri87,Dri90a}. Their connection with the functional version of Sklyanin's SOV was understood in \cite{Ter99}, where were constructed the factorizing $F$-matrices for arbitrary spin representations of the rational 6-vertex Yang-€"Baxter algebra.

Using the  SOV basis of Theorem~\ref{th-SOV-basis}, it is possible to completely characterize the spectrum and the eigenvectors of the ($\kappa$-twisted) antiperiodic transfer matrix $\bar{\mathsf{T}}^{(\kappa)}$.
As in \cite{Skl90,Skl92,Nic13,Nic13b}, the latter can be obtained in terms of the solutions of a discrete system of equations:

\begin{theorem}
\label{th-eigen-t}
For any fixed $\kappa\in\mathbb{C}\setminus\{0\}$, the $\kappa$-twisted antiperiodic tranfer matrix $\bar{\mathsf{T}}^{(\kappa)}(\lambda)$ \eqref{anti-T} defines a one-parameter family of commuting operators.
All these families are isospectral, i.e. they admit the same set of eigenvalue functions $\Sigma_{\mathsf{\bar{T}}}$ independently of the value of $\kappa$.

Under the conditions \eqref{cond-inh}, the spectrum of $\bar{\mathsf{T}}^{(\kappa)}(\lambda )$ is simple and $\Sigma _{{\mathsf{\bar{T}}}}$ coincides with the set of functions of the form 
\begin{equation}\label{t-form}
   \bar{\mathsf{t}}(\lambda)=\sum_{n=1}^\mathsf{N} 
   \prod_{\substack{\ell=1\\ \ell\not=n}}^\mathsf{N}\frac{\sinh(\lambda-\xi_\ell)}{\sinh(\xi_n-\xi_\ell)}\, \bar{\mathsf{t}}(\xi_n),
\end{equation}
satisfying the discrete system of equations:
\begin{equation}
\det_{2s_{n}+1}D_{\bar{\mathsf{t}}}^{(n)}=0,\qquad\forall n\in \{1,\ldots,\mathsf{N}\},
\label{discrete-sys}
\end{equation}
where $D_{\bar{\mathsf{t}}}^{(n)}$ is the $( 2s_{n}+1) \times( 2s_{n}+1) $ tridiagonal matrix
\begin{equation}\label{Dt}
   D_{\bar{\mathsf{t}}}^{(n)}\equiv 
\begin{pmatrix}
\bar{\mathsf{t}}(\xi_{n}^{(0)}) & a(\xi_{n}^{(0)}) & 0 &  \hdotsfor{2} & 0 \\ 
-d(\xi_{n}^{(1)}) & \bar{\mathsf{t}}(\xi_{n}^{(1)}) & a(\xi_{n}^{(1)}) & 0  & & \vdots  \\ 
0 &  & \ddots &  &    &  \vdots \\ 
\vdots &  &  & \ddots &   &  0 \\ 
\vdots &  & 0 & -d(\xi_{n}^{(2s_{n}-1)}) & \bar{\mathsf{t}}(\xi_{n}^{(2s_{n}-1)}) & a(\xi_{n}^{(2s_{n}-1)}) \\ 
0 & \hdotsfor{2}    & 0 & -d(\xi_{n}^{(2s_{n})}) & \bar{\mathsf{t}}(\xi_{n}^{(2s_{n})})
\end{pmatrix}.
\end{equation}
The left $\mathsf{\bar{T}}^{(\kappa)}(\lambda)$-eigenstate $\bra{\bar{\mathsf{t}}^{(\kappa)} }$ and the right $\mathsf{\bar{T}}^{(\kappa)}(\lambda)$-eigenstate $\ket{ \bar{\mathsf{t}}^{(\kappa)} }$ corresponding to the eigenvalue $\bar{\mathsf{t}}(\lambda )\in \Sigma _{\mathsf{\bar{T}}}$ are
\begin{align}
  &\bra{\bar{\mathsf{t}}^{(\kappa)} } 
  = \sum_{\mathbf{h}}
     \prod_{n=1}^\mathsf{N} \left\{\kappa^{h_n}\, \mathsf{q}_{\,\bar{\mathsf{t}},h_n}^{(n)} \right\}\,
     \prod_{\substack{i,j=1 \\ i<j}}^\mathsf{N}\sinh(\xi_j^{(h_j)}-\xi_i^{(h_i)})\, \bra{\mathbf{h} },
     \label{eigenl}\\
   &\ket{\bar{\mathsf{t}}^{(\kappa)} } 
  = \sum_{\mathbf{h}}
     \prod_{n=1}^\mathsf{N} \left\{\kappa^{-h_n}\, {\mathsf{p}}_{\,\bar{\mathsf{t}},h_n}^{(n)} \right\}
     \prod_{\substack{i,j=1 \\ i<j}}^\mathsf{N}\sinh(\xi_j^{(h_j)}-\xi_i^{(h_i)})\, \ket{\mathbf{h} }.  \label{eigenr}
\end{align}
In \eqref{eigenl}, \eqref{eigenr}, the summation runs over all $\mathsf{N}$-tuples $\mathbf{h}=(h_1,\ldots,h_\mathsf{N})$ such that $h_n\in\{0,1,\ldots,2s_n\}$ and, for each $n\in\{1,\ldots,\mathsf{N}\}$ and $h_n\in\{0,1,\ldots 2s_n\}$,  the coefficient $ \mathsf{q}_{\,\bar{\mathsf{t}},h_n}^{(n)}$ (resp. ${\mathsf{p}}_{\,\bar{\mathsf{t}},h_n}^{(n)}$) stands for the $(h_n+1)$-th coordinate of the unique (up to an overall normalization) vector $ \mathbf{q}_{\,\bar{\mathsf{t}}}^{(n)}$ (resp. ${\mathbf{p}}_{\,\bar{\mathsf{t}}}^{(n)}$) spanning the one-dimensional right (resp. left) nullspace of the matrix $D_{\bar{\mathsf{t}}}^{(n)}$.
\end{theorem}

\begin{rem}\label{rem-eigenvector-D}
The  two minors of order $2s_n$ associated respectively with the elements $(2s_n+1,1)$ and $(1,2s_n+1)$ of the matrix $D_{\bar{\mathsf{t}}}^{(n)}$ \eqref{Dt} are both obviously non-zero. Hence, the nullspace of $D_{\bar{\mathsf{t}}}^{(n)}$ satisfying \eqref{discrete-sys} is exactly of dimension 1.
The right nullspace is generated by a vector 
\begin{equation}\label{q-vect}
   \mathbf{q}_{\,\bar{\mathsf{t}}}^{(n)}
= \begin{pmatrix}
   \mathsf{q}_{\,\bar{\mathsf{t}},0}^{(n)} \\ \mathsf{q}_{\,\bar{\mathsf{t}},1}^{(n)} \\ 
   \vdots \\ \mathsf{q}_{\,\bar{\mathsf{t}},2s_n}^{(n)} 
   \end{pmatrix},
   \qquad\text{such that}\quad
   D_{\bar{\mathsf{t}}}^{(n)}\cdot \mathbf{q}_{\,\bar{\mathsf{t}}}^{(n)} =0,
\end{equation}
%
which is used to construct the left eigenstate \eqref{eigenl}. Note that, due to the non-vanishing of the aforementioned minors of $D_{\bar{\mathsf{t}}}^{(n)}$, we have both $\mathsf{q}_{\,\bar{\mathsf{t}},0}^{(n)}\not=0$ and $\mathsf{q}_{\,\bar{\mathsf{t}},2s_n}^{(n)}\not=0$ (so that we can for instance choose $\mathsf{q}_{\,\bar{\mathsf{t}},0}^{(n)}=1$). The coordinates $\mathsf{q}_{\,\bar{\mathsf{t}},h}^{(n)}$, $2\le h \le 2s_n$, of the vector \eqref{q-vect} can simply be obtained from the first one  $\mathsf{q}_{\,\bar{\mathsf{t}},0}^{(n)}$ by means of the following recursion relation:
\begin{equation}\label{rec-y}
 \mathsf{q}_{\,\bar{\mathsf{t}},h+1}^{(n)}
 =-\frac{\bar{\mathsf{t}}(\xi_n^{(h)})}{a(\xi_n^{(h)})}\, \mathsf{q}_{\,\bar{\mathsf{t}},h}^{(n)}
 +\frac{d(\xi_n^{(h)})}{a(\xi_n^{(h)})}\,  \mathsf{q}_{\,\bar{\mathsf{t}},h-1}^{(n)},
 \quad\text{for}\quad 0\le h\le 2s_n-1,
\end{equation}
where we have set $ \mathsf{q}_{\,\bar{\mathsf{t}},-1}^{(n)}=0$.
%
In the same way, the left nullspace of $D_{\bar{\mathsf{t}}}^{(n)}$ is generated by a vector
\begin{equation}\label{p-vect}
    \mathbf{p}_{\,\bar{\mathsf{t}}}^{(n)}\equiv
    \begin{pmatrix}\mathsf{p}_{\,\bar{\mathsf{t}},0}^{(n)}\\ \mathsf{p}_{\,\bar{\mathsf{t}},1}^{(n)}\\ \vdots \\\mathsf{p}_{\,\bar{\mathsf{t}},2s_n}^{(n)}
    \end{pmatrix},
   \qquad\text{such that}\quad
   {}^t \mathbf{p}_{\,\bar{\mathsf{t}}}^{(n)} \cdot D_{\bar{\mathsf{t}}}^{(n)} =0,
\end{equation}
which can be related to $\mathbf{q}_{\,\bar{\mathsf{t}}}^{(n)}$ as
\begin{equation}\label{left-vector}
   \mathsf{p}_{\,\bar{\mathsf{t}},h}^{(n)}
   =(-1)^h\left(\prod_{\ell=0}^{h-1}\!\frac{a(\xi_n^{(\ell)})}{d(\xi_n^{(\ell+1)})}\right)
   \mathsf{q}_{\,\bar{\mathsf{t}},h}^{(n)},
   \qquad
   0\le h\le 2s_n.
\end{equation}
This is due to the fact that ${}^t\! D_{\bar{\mathsf{t}}}^{(n)}$ can be obtained from $D_{\bar{\mathsf{t}}}^{(n)}$ by similarity by the diagonal matrix $X^{(n)}$ with elements
\begin{equation}
   \big[ X^{(n)}\big]_{j,k}=\delta_{j,k}\, \prod_{\ell=0}^{j-2}\left(-\frac{a(\xi_n^{(\ell)})}{d(\xi_n^{(\ell+1)})}\right),
   \qquad 1\le j,k\le 2s_n+1.
\end{equation}
Hence, in fact, both eigenvectors \eqref{eigenl}-\eqref{eigenr} can be written in terms of the vector $\mathbf{q}_{\,\bar{\mathsf{t}}}^{(n)}$ \eqref{q-vect} solution of \eqref{rec-y}.
\end{rem}

\begin{proof}
The proof of this theorem relies on the explicit construction of the transfer matrix eigenvectors $\bra{\bar{\mathsf{t}}^{(\kappa)} } $ (or $\ket{\bar{\mathsf{t}}^{(\kappa)} } $ ) in the SOV basis \eqref{leigen-h} (or \eqref{reigen-h}). From the computation of the matrix element $\bra{\bar{\mathsf{t}}^{(\kappa)} } \bar{\mathsf{T}}^{(\kappa)}(\lambda ) \ket{\mathbf{h}}$ by acting with the transfer matrix evaluated  in the  point $\lambda=\xi_n^{(h_n)}$, $1\le n\le \mathsf{N}$, $0\le h_n\le 2s_n$, on the right and on the left, one obtains a discrete system of equations for the wave-functions $\psi_{\bar{\mathsf{t}}^{(\kappa)} }(\mathbf{h})\equiv \bra{\bar{\mathsf{t}}^{(\kappa)} }\, \mathbf{h}\,\rangle$, which corresponds to a union of homogeneous systems with matrices \eqref{Dt}. The condition for this system to admit a non-zero solution is precisely \eqref{discrete-sys}, which does not depend on $\kappa$. 
The simplicity of the spectrum is then related to the fact that the nullspace of the matrix $D_{\bar{\mathsf{t}}}^{(n)}$ \eqref{Dt} satisfying \eqref{discrete-sys} is exactly of dimension 1, which also means that  the eigenvectors are, up to an overall normalization, given in terms of the vectors \eqref{q-vect} and \eqref{p-vect} as in \eqref{eigenl}-\eqref{eq-left}.
We refer the reader to \cite{Nic13,Nic13b} for more details on the different steps of the proof.
\end{proof}

The SOV discrete sets of equations \eqref{discrete-sys}, however, although interesting for numerically characterizing the transfer matrix spectrum of small-size systems, are inadequate in their present form for the study of large-size systems and the consideration of the homogeneous limit of the model.
Therefore, as announced in Introduction, we shall in the next sections discuss its reformulation in terms of {\em continuous} functional equations of Baxter's type, i.e. in terms of systems of Bethe-type equations.

As already noticed by Sklyanin \cite{Skl90,Skl92}, the set of $\mathsf{N}$ finite-difference discrete systems  characterizing the SOV spectrum of the antiperiodic transfer matrix (i.e. conditioning the existence of a non-zero vector within the nullspace of each of the matrices \eqref{Dt}) and leading to the explicit construction of its eigenvectors (see \eqref{rec-y}) are identical, in their form, to some discrete version of Baxter's famous $T$-$Q$ second-order finite-difference functional equation involving the eigenvalue of the transfer matrix together with some convenient (entire) function $Q(\lambda)$ \cite{Bax72,Bax82L}.
More precisely, if for some entire function ${\bar{\mathsf{t}}}(\lambda)$ of the form \eqref{t-form}, we find a function $Q(\lambda)$ such that the system of equations
\begin{equation}
 {\bar{\mathsf{t}}}(\xi_n^{(h_n)})\, Q (\xi_n^{(h_n)})
  =- a(\xi_n^{(h_n)})\, Q (\xi_n^{(h_n+1)})
    + d(\xi_n^{(h_n)})\, Q (\xi_n^{(h_n-1)}),
    \quad
    h_n\in\{0,1\ldots,2s_n\},
   \label{eq-left}
\end{equation}
is satisfied for each $n\in\{1,\ldots,\mathsf{N}\}$, it means that, for each $n\in\{1,\ldots,\mathsf{N}\}$, the  vector
\begin{equation}\label{def-Qn}
  \mathbf{Q}^{(n)}\equiv
  \begin{pmatrix}
  Q(\xi_n^{(0)}) \\ Q(\xi_n^{(1)}) \\ \vdots \\ Q(\xi_n^{(2s_n)})
  \end{pmatrix}
\end{equation}
belongs to the nullspace of the matrix $D_{\bar{\mathsf{t}}}^{(n)}$ \eqref{Dt}.
If moreover all the vectors \eqref{def-Qn} are non-zero, then \eqref{discrete-sys} is satisfied and ${\bar{\mathsf{t}}}(\lambda)\in\Sigma_{\bar{\mathsf{T}}}$. Each of the vectors $\mathbf{Q}^{(n)}$ \eqref{def-Qn} being then proportional to the corresponding vector $\mathbf{q}_{\,\bar{\mathsf{t}}}^{(n)}$ \eqref{q-vect}, they can therefore be directly used for the construction of the eigenvectors \eqref{eigenl} and \eqref{eigenr}.

In general, for a given ${\bar{\mathsf{t}}}(\lambda)\in\Sigma_{\bar{\mathsf{T}}}$, one can exhibit many such functions $Q(\lambda)$: it is enough that the latter interpolates, at each point $\xi_n^{(h_n)}$, the coordinate $\mathsf{q}_{\bar{\,\mathsf{t}},h_n}^{(n)}$ (up to a possible proportionality constant $c_n$) of the vector \eqref{q-vect}.
Solutions of Baxter's homogeneous $T$-$Q$ functional equation (i.e. of \eqref{hom-TQ} where, in the present case, $\mathsf{t}(\lambda)$, $\mathsf{a}(\lambda)$, $\mathsf{d}(\lambda)$ and $\epsilon_q$ are respectively identified with $\bar{\mathsf{t}}(\lambda)$,  $-a(\lambda)$, $d(\lambda)$ and $-\eta$) provide such functions $Q(\lambda)$ (at the condition that these solutions do not vanish on a whole string of values $\xi_n^{(0)},\xi_n^{(1)},\ldots,\xi_n^{(2s_n)}$).
If one is able to clearly identify their functional form\footnote{which is for the moment not the case for all integrable models, especially for problems with general integrable boundary conditions.}, they are probably the most convenient ones for the writing of a system of Bethe equations.
They are notwithstanding not the only ones.
It was notably proposed recently, in a series of papers \cite{CaoYSW13a,CaoYSW13b,CaoCYSW14,LiCYSW14,CaoYSW14}, that one could as well try to characterize the transfer matrix spectrum by means of solutions of some inhomogeneous version of Baxter's functional equation. The advantage of these formulations
is that one has in principle more freedom to fix the functional form of the corresponding solution $Q(\lambda)$. One can therefore choose to consider only solutions with the same functional form as the functions defining the model (i.e., in the case of the XXZ chain, functions of the form \eqref{Q-2pi}), independently of its boundary conditions.
Hence, before discussing in Section~\ref{sec-hom} the reformulation of Theorem~\ref{th-eigen-t} in terms of the solutions to the usual (homogeneous) $T$-$Q$ functional equation and the completeness of this characterization, we find it instructive, in Section~\ref{sec-func}, to also discuss these problems for the  solutions to some higher-spin generalization of the  inhomogeneous functional equations recently proposed in \cite{ZhaLCYSW14}.
Note at this point that a similar reformulation of the SOV spectrum characterization has already been obtained in different situations, in terms of homogeneous $T$-$Q$ functional equations in \cite{NicT10,Nic10a,Nic11,GroN12} for cyclic representations of the trigonometric 6-vertex Yang-Baxter algebra, and in terms of inhomogeneous  $T$-$Q$ functional equations in \cite{KitMN14} for open XXZ spin-1/2 chains.



\section{Reformulation of the SOV spectrum in terms of solutions of an inhomogeneous functional equation}
\label{sec-func}

As already mentioned, inhomogeneous versions of Baxter's $T$-$Q$ functional equations appeared first in the literature in the context of the so-called {\em off-diagonal Bethe ansatz} (ODBA) in \cite{CaoYSW13a}, as an ansatz to characterize the eigenvalues\footnote{A construction of the eigenstates has also been recently proposed in \cite{ZhaLCYSW14} for the spin-1/2 antiperiodic XXZ chain and the open XXX chain. Let us remark that the approach described in \cite{ZhaLCYSW14} requires the existence of a SOV or of an $F$-basis \cite{MaiS00}. In fact, the eigenstates presented there just coincide with those previously constructed by SOV in \cite{Nic13} and \cite{Nic12,FalKN14}. It is nevertheless worth noticing that  \cite{ZhaLCYSW14} provides an interesting rewriting of theses eigenstates, the interest being in the apparent possibility to take the homogeneous limit without encountering fictitious divergences, as shown for small chains in \cite{ZhaLCYSW14}.} of the antiperiodic XXZ spin-1/2 chains (other quantum integrable models have then been considered in this framework, see for instance \cite{CaoYSW13b,CaoCYSW14,LiCYSW14}).
Various inhomogeneous $T$-$Q$ functional equations have there been proposed to meet the following criterium: when particularized at some special points (the inhomogeneity parameters of the model) they induce for their $T$-solutions a set of discrete identities which are known to be satisfied by all transfer matrix eigenvalues.
As already mentioned, the initial idea leading to the addition of an inhomogeneous term in the $T$-$Q$ equation comes from the desire to allow solutions of the form \eqref{Q-2pi} which, in the antiperiodic case, obviously do not satisfy the usual homogeneous $T$-$Q$ equation. This could be {\em a priori} interesting for situations for which the functional form of the $Q$-operator eigenvalues is not known, in particular for chains with non-diagonal boundary conditions (see also \cite{BelC13} where such inhomogeneous equations appear in another context).
A natural question is therefore whether, when one considers solutions of the form \eqref{Q-2pi} (with some eventual additional requirement on the degree) of some given inhomogeneous $T$-$Q$ equation, one obtains a {\em full characterization} of the transfer matrix spectrum. 

To answer this question, one has in principle to demonstrate two main statements:
\begin{itemize}[topsep=0.5pt,itemsep=-0.5ex,partopsep=1ex,parsep=1ex]
\item[{\sf (a)}] that the set of solutions $\Sigma_{\mathsf{T}}'$ to the aforementioned system of discrete identities (together with some additional restrictions on the functional form of the solution) exactly coincides  with the set $\Sigma_{\bar{\mathsf{T}}}$ of the transfer matrix eigenvalues;
\item[{\sf (b)}] that the set of the $T$-solutions  $\Sigma_{\mathsf{T}}''$ to the proposed inhomogeneous $T$-$Q$ functional equation exactly coincides with $\Sigma_{\mathsf{T}}'$.
\end{itemize}
Indeed, unless these two points are proven, one can neither state that the inhomogeneous $T$-$Q$ equation describes the complete set of transfer matrix eigenvalues, nor that each $T$-solution effectively corresponds to some transfer matrix eigenvalue.
Let us remark at this point that the discrete identities in question appear in the ODBA framework only as a set of {\em necessary conditions} that should be fulfilled by all transfer matrix eigenvalues\footnote{They are just, in the spin-1/2 case, the consequence on the eigenvalues of the inversion formula for the transfer matrix, known since the papers \cite{KulS82,IzeK82}, while for more general representations they coincides with the particularization at the inhomogeneity parameters of the well known fusion identities  \cite{KulS82,Res83,KirR86a,KirR87}.}, so that this approach ensures only the inclusion  $\Sigma_{\bar{\mathsf{T}}}\subset\Sigma_{\mathsf{T}}'$.
Indeed, without further studies, the knowledge on the spectrum provided by ODBA does not allow one to guarantee  that the latter is completely characterized by these conditions\footnote{In particular, we would like to point out that the arguments presented in the recent preprint \textsf{arXiv:1409.5303} are incomplete. There, the authors pretend to have proven point {\sf (a)} just by comparing (and this is the heart of their argumentation) the maximum number of solutions to the aforementioned set of discrete identities with the dimension of the space of states (which happen to be the same). By making this comparison, the authors implicitly assume that the number of independent eigenvalues of the corresponding transfer matrix exactly coincides with the dimension of the Hilbert space. In other words, they implicitly assume {\sf (1)} that the transfer matrix is diagonalizable, and {\sf (2)} that its spectrum is non-degenerate. No proof of these crucial properties are provided within the ODBA approach. Whereas the first one may be deduced from simple considerations such as in Proposition~\ref{prop-normal}, the proof of the second one requires instead more advanced arguments, such as those provided by SOV for generic inhomogeneous models. Let us by the way point out that, for the two examples explicitly considered in that paper, i.e. for spin 1/2 inhomogeneous XXZ chains with antiperiodic and general open boundary conditions, the statements of their Theorem 1 and Corollary 2 had been previously proven in the SOV framework respectively in \cite{Nic13} and \cite{Nic12,FalKN14} ,whereas the proof of Corollary 4 was presented in \cite{KitMN14}.}.
On the contrary, the SOV approach enables us to affirm the {\em simplicity} of the transfer matrix spectrum, a crucial property will leads to its complete characterization and therefore provides a proof of statement {\sf (a)}: this is precisely the content of Theorem~\ref{th-eigen-t}, the discrete identities used in \cite{CaoYSW13a} coinciding with the spin-1/2 particularization of \eqref{discrete-sys}.
The aim of the present section is now to see whether one can also prove point {\sf (b)}, for a one-parameter family of ``minimal" inhomogeneous $T$-$Q$ functional equations generalizing one of the equations recently proposed in \cite{ZhaLCYSW14}.
Our result can be stated as follows.

\begin{theorem}\label{th-inh-eq}
Under the condition \eqref{cond-inh}, the following propositions are equivalent:
\begin{enumerate}
   \item\label{cond1}
   $\mathsf{\bar{t}}(\lambda)$ is an eigenvalue function of the $\kappa$-twisted antiperiodic transfer matrix $\bar{\mathsf{T}}^{(\kappa)}(\lambda)$ for any $\kappa\in\mathbb{C}\setminus\{0\}$
   (i.e. $\mathsf{\bar{t}}(\lambda)\in\Sigma_{\bar{\mathsf{T}}}$).
   \item\label{cond2}
   $\mathsf{\bar{t}}(\lambda)$ is an entire function of $\lambda$ and, for some $\alpha\in\mathbb{C}$, there exists a function $Q(\lambda)$ of the form
   \begin{equation}\label{Q-Ns}
      Q(\lambda)=\prod_{j=1}^{\mathsf{N}_s}\sinh (\lambda -\lambda_j),
      \qquad
      \lambda_1,\ldots,\lambda_{\mathsf{N}_s}\in\mathbb{C},
      \qquad
      \mathsf{N}_s\equiv 2\sum_{a=1}^{\mathsf{N}} s_{a},
   \end{equation}
   such that $Q(\xi_j^{(0)})\not=0$ for $1\le j\le \mathsf{N}$, and such that $\mathsf{\bar{t}}(\lambda)$ and $Q(\lambda)$ satisfy the following inhomogeneous finite-difference functional equation:
   \begin{equation}\label{eq-inh}
      \mathsf{\bar{t}}(\lambda)\, Q(\lambda)
      =-e^{\lambda-\alpha}\, a(\lambda)\, Q(\lambda-\eta)
      +e^{-\lambda-\eta+\alpha}\, d(\lambda)\, Q(\lambda+\eta)
      +  F_{\alpha+\bar\lambda}(\lambda),
   \end{equation}
   where $\bar\lambda$ stands for the sum of the roots $\bar{\lambda}\equiv\sum_{j=1}^{\mathsf{N}_s}\lambda_j$ of the function $Q(\lambda)$. The function $F_x(\lambda)$ is defined as
   \begin{equation}\label{F-Q}
     F_x(\lambda)=2 e^{-\frac{\mathsf{N}_s+1}{2}\eta}
     \sinh\!\Bigg(
     \lambda-x +\sum_{j=1}^\mathsf{N}\sum_{h_j=1}^{2s_j}\xi_j^{(h_j)}+\frac{\mathsf{N}_s+1}{2}\eta\!
     \Bigg) 
     \prod_{j=1}^\mathsf{N}\prod_{h_j=0}^{2s_j}\!\sinh\!\big(\lambda-\xi_j^{(h_j)}\big) .
   \end{equation}
   \item\label{cond3}
   $\mathsf{\bar{t}}(\lambda)$ is an entire function of $\lambda$ and, for any $\alpha\in\mathbb{C}$, there exists a function $Q(\lambda)$ of the form \eqref{Q-Ns} such that $\mathsf{\bar{t}}(\lambda)$ and $Q(\lambda)$ satisfy the inhomogeneous functional equation \eqref{eq-inh}-\eqref{F-Q}.
For all but at most finitely many values of $e^\alpha$, the corresponding solution is unique. For all but finitely many values of $e^\alpha$, it is also such that  $Q(\xi_j^{(0)})\not=0$ for $1\le j\le \mathsf{N}$.
\end{enumerate}
\end{theorem}

\begin{rem}
Note that, contrary to what happens in the periodic case (see for instance \cite{TarV95}), the deformation parameter $\alpha$ is not related to the twist parameter $\kappa$. It is only introduced here for convenience.
\end{rem}

\begin{rem}\label{rem-xi_j}
The existence of a non-zero solution $Q(\lambda)$ of the form \eqref{Q-Ns} to \eqref{eq-inh} for $\bar{\mathsf{t}}(\lambda)$ entire is {\em a priori} not enough to ensure that $\bar{\mathsf{t}}(\lambda)$ belongs to $\Sigma_{\bar{\mathsf{T}}}$: in the case where one of the vectors \eqref{def-Qn} vanishes, it does not provide us with a non-zero element of the nullspace of the corresponding matrix \eqref{Dt}, and therefore does not imply the vanishing of the corresponding determinant.
This explains why we introduce in Theorem~\ref{th-inh-eq} an additional condition, namely the fact that  $Q(\xi_j^{(0)})\not=0$ for all $j\in\{1,\ldots,\mathsf{N}\}$. The latter is indeed a  sufficient condition to ensure that the corresponding $\bar{\mathsf{t}}(\lambda)$ satisfies the discrete system  \eqref{discrete-sys}.
Such solutions will be called {\em admissible}.
\end{rem}

This theorem enables us to completely characterize the spectrum and eigenstates of the ($\kappa$-twisted) antiperiodic transfer matrix in terms of the admissible solutions $\Lambda_\alpha=\{\lambda_1,\ldots,\lambda_{\mathsf{N}_s}\}$ of the system of generalized Bethe equations following from the cancellation conditions of all apparent poles at $\lambda=\lambda_j$, $1\le j\le \mathsf{N}_s$,  of the function
\begin{multline}\label{bethe-inh}
\bar{\mathsf{t}}_{\Lambda_\alpha}(\lambda)
=-e^{\lambda-\alpha} \, a(\lambda)\, \prod_{j=1}^{\mathsf{N}_s}\frac{\sinh(\lambda-\lambda_j-\eta)}{\sinh(\lambda-\lambda_j)}\\
   +e^{\alpha-\lambda-\eta}\,d(\lambda)\,\prod_{j=1}^{\mathsf{N}_s}\frac{\sinh(\lambda-\lambda_j+\eta)}{\sinh(\lambda-\lambda_j)}
   +\frac{F_{\alpha+\bar\lambda}(\lambda)}{\prod_{j=1}^{\mathsf{N}_s}\sinh(\lambda-\lambda_j)},
\end{multline}
where $F_{\alpha+\bar\lambda}$ is defined by \eqref{F-Q} in terms of $\bar{\lambda}\equiv\sum_{j=1}^{\mathsf{N}_s}\lambda_j$.

\begin{corollary}\label{cor-Bethe}
Under the condition \eqref{cond-inh}, for all but finitely many values of $e^\alpha$, there exists a one-to-one correspondence between $\Sigma_{\bar{\mathsf{T}}}$ and the set $\Sigma^{(\alpha)}_{\mathrm{InBAE}}$ of solutions $\Lambda_\alpha=\{\lambda_1,\ldots,\lambda_{\mathsf{N}_s}\}$ 
such that
\begin{equation}
\forall j\in\{1,\ldots,\mathsf{N}_s\},\qquad
\Im\lambda_j\in[0,\pi[ \quad\text{and}\quad
\lambda_j\not=\xi_k^{(0)}\hspace{-3mm}\mod i\pi,\ \, 1\le k\le \mathsf{N},
\end{equation}
to the system of $\mathsf{N}_s$ inhomogeneous equations of Bethe type following from the condition that the function \eqref{bethe-inh} is entire.
The eigenvalue $\bar{\mathsf{t}}_{\Lambda_\alpha}(\lambda)\in\Sigma_{\bar{\mathsf{T}}}$ associated to $\Lambda_\alpha\in\Sigma^{(\alpha)}_{\mathrm{InBAE}}$ is then given by the entire function \eqref{bethe-inh}.
For any $\kappa\in\mathbb{C}\setminus\{0\}$, the corresponding left and right $\bar{\mathsf{T}}^{(\kappa)}$-eigenstates are respectively given by
\begin{align}
  &\bra{ {\Lambda_\alpha} } 
  = \sum_{\mathbf{h}}
     \prod_{j=1}^\mathsf{N} \left\{\kappa^{h_j}\, Q_{\Lambda_\alpha}(\xi_j^{(h_j)})\right\}\,
     \prod_{i<j}\sinh(\xi_j^{(h_j)}-\xi_i^{(h_i)})\, \bra{\mathbf{h} },
     \label{eigenl-bis}\\
   &\ket{ {\Lambda_\alpha} } 
  = \sum_{\mathbf{h}}
     \prod_{j=1}^\mathsf{N} \left\{(-\kappa)^{-h_j}\!
     \left(\prod_{\ell=0}^{h_j-1}\!\frac{a(\xi_n^{(\ell)})}{d(\xi_n^{(\ell+1)})}\right)\!
     Q_{\Lambda_\alpha}(\xi_j^{(h_j)})\right\}
     \prod_{i<j}\sinh(\xi_j^{(h_j)}-\xi_i^{(h_i)})\, \ket{\mathbf{h} },  \label{eigenr-bis}
\end{align}
in terms of a function $Q_{\Lambda_\alpha}(\lambda)$ satisfying \eqref{eq-left}, and which can be expressed in terms of the set of Bethe roots $\Lambda_\alpha$ as follows:
\begin{equation}\label{Q-t}
   Q_{\Lambda_\alpha}(\lambda)=e^{-\frac{\lambda(\lambda+\eta-2\alpha)}{2\eta}}
   \prod_{j=1}^{\mathsf{N}_s}\sinh(\lambda-\lambda_j).
\end{equation}
\end{corollary}

The rest of this section is devoted to the proof of Theorem~\ref{th-inh-eq}. We shall for convenience adopt the terminology of Appendix~\ref{app-trig-pol}.

Before turning explicitly to the proof of Theorem~\ref{th-inh-eq}, lets us first formulate a few useful lemmas.

\begin{lemma}\label{lem-hom}
Let $\alpha\in\mathbb{C}$ and let $\mathsf{\bar{t}}(\lambda)$ be a function of the form \eqref{t-form}.
Then the homogeneous functional equation
\begin{equation}\label{hom-fct}
     \mathsf{\bar{t}}(\lambda)\, Q(\lambda)
      =-e^{\lambda-\alpha}\, a(\lambda)\, Q(\lambda-\eta)
      +e^{-\lambda-\eta+\alpha}\, d(\lambda)\, Q(\lambda+\eta)
\end{equation}
does not admit non-trivial solutions $Q(\lambda )$ in the class $\mathbb{T}[\lambda]$ of trigonometric polynomials in $\lambda$.
\end{lemma}

\begin{proof}
It is enough to observe that, if $Q(\lambda)\in\mathbb{T}[\lambda]\equiv\mathbb{C}[e^\lambda,e^{-\lambda}]$, the r.h.s. and the l.h.s. of \eqref{hom-fct} have different leading asymptotic behaviors when $\lambda\to\pm\infty$.
\end{proof}

\begin{lemma}\label{lem-deg}
Under the hypothesis of Lemma~\ref{lem-hom},
let us suppose that $Q(\lambda)\in\mathbb{T}_{\mathsf{N}_Q}[\lambda]$ is  such that
\begin{equation}\label{fct-F}
     \mathsf{\bar{t}}(\lambda)\, Q(\lambda)
      =-e^{\lambda-\alpha}\, a(\lambda)\, Q(\lambda-\eta)
      +e^{-\lambda-\eta+\alpha}\, d(\lambda)\, Q(\lambda+\eta)+F(\lambda)
\end{equation}
for some function $F(\lambda)$ vanishing at $\mathsf{N}+\mathsf{N}_Q+1$ points which are pairwise distinct modulo $i\pi$. Then $Q(\lambda)\in\hat{\mathbb{T}}_{\mathsf{N}_Q}[\lambda]$ or $Q(\lambda)$ is identically zero.
\end{lemma}

\begin{proof}
Let us suppose that $Q(\lambda)\notin\hat{\mathbb{T}}_{\mathsf{N}_Q}[\lambda]$, i.e. that $Q(\lambda)$ can be written in the form
\begin{equation}
  Q(\lambda)=e^{-\mathsf{N}_Q\lambda}\sum_{j=0}^{\mathsf{N}_Q}c_j\, e^{2j\lambda},
\end{equation}
with either $c_{\mathsf{N}_Q}=0$ or $c_0=0$.
 It means that $Q(\lambda)\in {\mathbb{T}}_{\mathsf{N}',\mathsf{N}_Q-1}[\lambda]$, with $\mathsf{N}'=\mathsf{N}_Q$ or $\mathsf{N}'=\mathsf{N}_Q-2$.
Then, from \eqref{fct-F}, $F(\lambda)\in {\mathbb{T}}_{\mathsf{N}'+\mathsf{N}+1,\mathsf{N}_Q+\mathsf{N}}[\lambda]$. Since $F(\lambda)$ vanishes in $\mathsf{N}+\mathsf{N}_Q+1$ different points, we have $F(\lambda)=0$. It follows from Lemma~\ref{lem-hom} that $Q(\lambda)=0$.
\end{proof}

\begin{lemma}\label{lem-ZQ}
If $Q(\lambda)$ is a trigonometric polynomial of the form \eqref{Q-Ns}, then the function $Z_{\alpha,Q}(\lambda)$ defined as
\begin{equation}\label{Z-Q}
  Z_{\alpha,Q}(\lambda)= -e^{\lambda-\alpha} a(\lambda)\, Q(\lambda-\eta)
      +e^{-\lambda-\eta+\alpha} d(\lambda)\, Q(\lambda+\eta)
      +  F_{\alpha+\bar\lambda}(\lambda),
\end{equation}
where $F_{\alpha+\bar\lambda}(\lambda)$ is given by \eqref{F-Q}, is an element of $\mathbb{T}_{\mathsf{N}_s+\mathsf{N}-1}[\lambda]$.
\end{lemma}

\begin{proof}
It is straightforward to see that $Z_{\alpha,Q}(\lambda)$ is an element of $\mathbb{T}_{\mathsf{N}_s+\mathsf{N}+1}[\lambda]$. The claim follows from the fact that
\begin{equation}\label{asympt-Z}
  \lim_{\lambda \rightarrow \pm \infty }
  e^{\mp (\mathsf{N}_s+\mathsf{N}+1)\lambda}\, Z_{\alpha,Q}(\lambda) =0.
\end{equation}
\end{proof}

\smallskip
\noindent
\textit{Proof of Theorem~\ref{th-inh-eq}.}
The fact that {\it \ref{cond3}.} implies {\it \ref{cond2}.} is obvious.

The fact that  {\it \ref{cond2}.} implies  {\it \ref{cond1}.} is also quite straightforward. Indeed, let us suppose that   {\it \ref{cond2}.}  holds. 
It follows from \eqref{eq-inh}, from Lemma~\ref{lem-ZQ} and from the fact that $\bar{\mathsf{t}}(\lambda)$ is an entire function in  $\lambda$ that the latter is indeed of the form \eqref{t-form}.
Moreover, particularizing \eqref{eq-inh} in the $\mathsf{N}+\mathsf{N}_s$ points $\xi_n^{(h_n)}$, $1\le n\le\mathsf{N}$, $0\le h_n\le 2s_n$, we get that each of the $\mathsf{N}$ following homogeneous systems of equations ($1\le n\le \mathsf{N}$):
\begin{equation}
   \sum_{k=1}^{2s_n+1} \Big[ \big(X_\alpha\, X_f^{(n)}\big)^{-1}\, D_{\bar{\mathsf{t}}}^{(n)}\, 
   \big(X_{\alpha}\, X_f^{(n)}\big) \Big]_{j,k}\ \mathsf{v}_k^{(n)}=0,
   \qquad 1\le j \le 2s_n+1,
\end{equation}
admit non-zero solutions given by
\begin{equation}
  \mathsf{v}_j^{(n)}=Q(\xi_n^{(j-1)}), \qquad 1\le j\le 2s_n+1.
\end{equation}
Here $D_{\bar{\mathsf{t}}}^{(n)}$ is the $(2s_n+1)\times(2s_n+1)$ matrix \eqref{Dt}, and $X_{\alpha}$, $X_f^{(n)}$ are the $(2s_n+1)\times(2s_n+1)$ diagonal matrices of elements
\begin{equation}\label{Xn}
 \big[X_{\alpha}\big]_{j,k}=\delta_{j,k}\, e^{-(j-1)\alpha},
 \qquad
  \big[X_f^{(n)}\big]_{j,k}=\delta_{j,k}\,\prod_{\ell=0}^{j-2}f(\xi_n^{(\ell)}), 
 \qquad 1\le j,k\le (2s_n+1),
\end{equation}
with here $f(\lambda)=e^\lambda$.
Hence it means that $\bar{\mathsf{t}}(\lambda)$ satisfies the system \eqref{discrete-sys}, and therefore that $\bar{\mathsf{t}}(\lambda)\in\Sigma_{\bar{\mathsf{T}}}$.

Finally, let us show that  {\it \ref{cond1}.} implies  {\it \ref{cond3}.}
Let us point out that if $\bar{\mathsf{t}}(\lambda)\in\Sigma_{\bar{\mathsf{T}}}$ and $Q(\lambda )$ is a trigonometric polynomial in $\lambda$ of the form \eqref{Q-Ns}, then, from Lemma~\ref{lem-ZQ}, both members of the functional equation \eqref{eq-inh} belong to $\mathbb{T}_{\mathsf{N}_s+\mathsf{N}-1}[\lambda]$.
Hence the equation \eqref{eq-inh} for $\bar{\mathsf{t}}(\lambda)$ and $Q(\lambda)$ is satisfied if and only if it is satisfied in $\mathsf{N}_s+\mathsf{N}$  different values of $\lambda$, say at the $\mathsf{N}_s+\mathsf{N}$ points $\xi_j^{(h_j)}$, $1\le j\le \mathsf{N}$, $0\le h_j\le 2s_j$.
In other words, it means that the functional equation \eqref{eq-inh} is equivalent to the following system of $\mathsf{N}$ sub-systems of the form (corresponding to a global system of $\mathsf{N}+\mathsf{N}_s$ equations):
\begin{equation}\label{sub-sys}
  \big(X_\alpha\, X_f^{(n)}\big)^{-1}\, D_{\bar{\mathsf{t}}}^{(n)}\, 
   \big(X_{\alpha}\, X_f^{(n)}\big)
\cdot
  \left( 
\begin{array}{c}
Q(\xi_n^{(0)}) \\ 
Q(\xi_n^{(1)}) \\ 
\vdots \\ 
Q(\xi_n^{(2s_{n})})
\end{array}
\right) 
=\left( 
\begin{array}{c}
0 \\ 
\vdots \\ 
\vdots \\ 
0
\end{array}
\right),
\qquad
\forall n\in \{1,\ldots,\mathsf{N}\},
\end{equation}
where $D_{\bar{\mathsf{t}}}^{(n)}$ is the $(2s_n+1)\times(2s_n+1)$ matrix \eqref{Dt} and $X_{\alpha}$, $X_f^{(n)}$ are given by \eqref{Xn} with $f(\lambda)=e^\lambda$.

We recall that each of these subsystems admit, up to an overall normalization coefficient $Q(\xi_n^{(0)})$, only one non-zero solution which is given in terms of the vector $\mathbf{q}_{\,\bar{\mathsf{t}}}^{(n)}$ \eqref{q-vect}-\eqref{rec-y} generating the nullspace of the matrix $D_{\bar{\mathsf{t}}}^{(n)}$ by
\begin{equation}
   Q(\xi_n^{(h)})= e^{h\alpha}\ \mathsf{x}_{\bar{\mathsf{t}},f,h}^{(n)}\ Q(\xi_n^{(0)}),
   \qquad \text{with} \quad
   \mathsf{x}_{\bar{\mathsf{t}},f,h}^{(n)}\equiv 
   \left[\prod_{\ell=0}^{h-1} f(\xi_n^{(\ell)})\right]^{-1}
   \mathsf{q}_{\,\bar{\mathsf{t}},h}^{(n)},
\end{equation}
for any $h\in\{0,\ldots 2s_n\}$.
We also recall that any trigonometric polynomial $Q(\lambda)\in\mathbb{T}_{\mathsf{N}_s}[\lambda]$ can be written as in \eqref{interpol-P} in terms of its values at $\mathsf{N}_s+1$ different points $\zeta_0,\zeta_1,\ldots,\zeta_{\mathsf{N}_s}$  (which are pairwise distinct modulo $i\pi$).
Choosing the points $\zeta_1,\ldots,\zeta_{\mathsf{N}_s}$ to be
\begin{equation}
  \zeta_{j_{n,h}}\equiv \xi_n^{(j)}, \quad
  \text{with}\quad j_{n,h}=2\sum_{k=1}^{n-1}s_k+h+1, \quad
  1\le n \le \mathsf{N}, \ \ 0\le h \le 2s_n-1,
\end{equation}
$\zeta_0$ being an arbitrary complex number (distinct from $\zeta_1,\ldots,\zeta_{\mathsf{N}_s}$ modulo $i\pi$), we obtain that the previous system of $\mathsf{N}+\mathsf{N}_s$ equations is equivalent, for $Q(\lambda)\in\mathbb{T}_{\mathsf{N}_s}[\lambda]$, to the union of the following coupled sub-systems, for $1\le n\le \mathsf{N}$:
\begin{equation}\label{sub-sys2}
  \left\{
  \begin{split}
   &Q(\xi_n^{(h)})=e^{h\alpha}\ \mathsf{x}_{\bar{\mathsf{t}},f,h}^{(n)}\ Q(\xi_n^{(0)}),
     \qquad 1\le h\le 2s_n,\\
   & Q(\xi_n^{(2s_n)})
     =\sum_{k=0}^{\mathsf{N_s}}\ \prod_{\substack{ \ell=0  \\ \ell\neq k }}^{\mathsf{N}_s}
     \frac{\sinh (\xi_n^{(2s_n)} -\zeta_{\ell})}{\sinh (\zeta_k-\zeta_\ell)} \ Q(\zeta_k).
  \end{split}
  \right.
\end{equation}
This system admits a non-zero solution if and only if the following homogeneous linear system of $\mathsf{N}$ equations,
\begin{multline}\label{sys-0-1}
   e^{2s_n\alpha\,}\mathsf{x}_{\bar{\mathsf{t}},f,2s_n}^{(n)}\  q_n
   -\sum_{j=1}^\mathsf{N}\!\sum_{h=0}^{2s_j-1}
   \frac{\sinh(\xi_n^{(2s_n)}-\zeta_0)}{\sinh(\xi_j^{(h)}-\zeta_0)}
   \underset{(k,\ell)\not=(j,h)\ }{\prod_{k=1}^\mathsf{N}\!\prod_{\ell=0}^{2s_k-1}}\!
   \frac{\sinh(\xi_n^{(2s_n)}-\xi_k^{(\ell)})}{\sinh(\xi_j^{(h)}-\xi_k^{(\ell)})}\,
   e^{j\alpha\,}\mathsf{x}_{\bar{\mathsf{t}},f,h}^{(j)}\ q_j
   \\
    - \prod_{\ell=1}^{\mathsf{N}_s}\frac{\sinh(\xi_n^{(2s_n)}-\zeta_\ell)}{\sinh(\zeta_0-\zeta_\ell)}\    q_0
    =0,
   \qquad
   1\le n\le \mathsf{N},
\end{multline}
in the $\mathsf{N}+1$ variables $q_0\equiv Q(\zeta_0)$, $q_j\equiv Q(\xi_j^{(0)})$, $1\le j\le \mathsf{N}$,
admits a non-zero solution.
It is obvious that it is always the case since the space of solutions to the system \eqref{sys-0-1} is at least of dimension 1. 
Therefore, for any value of $e^\alpha$, there exists a non-zero solution $Q(\lambda)\in\mathbb{T}_{\mathsf{N}_s}[\lambda]$ to the functional equation \eqref{eq-inh}. It follows from Lemma~\ref{lem-deg} that  $Q(\lambda)\in\hat{\mathbb{T}}_{\mathsf{N}_s}[\lambda]$, and hence is of the form \eqref{Q-Ns}.

The system \eqref{sys-0-1} can be rewritten as
\begin{equation}\label{sys-0-2}
   \sum_{j=1}^\mathsf{N} \left[\mathcal{M}_{\bar{\mathsf{t}}}(e^\alpha,\zeta_0)\right]_{i,j}\,  q_j
   = \prod_{\ell=1}^{\mathsf{N}_s}\frac{\sinh(\xi_i^{(2s_i)}-\zeta_\ell)}{\sinh(\zeta_0-\zeta_\ell)}\  q_0,
   \qquad
   1\le i\le \mathsf{N},
\end{equation}
where the matrix $\mathcal{M}_{\bar{\mathsf{t}}}(\beta,\zeta)$ is defined as
\begin{multline}\label{mat-sys0}
  \left[\mathcal{M}_{\bar{\mathsf{t}}}(\beta,\zeta)\right]_{i,j}
  =\delta_{i,j}\  \beta^{2s_i}\, \mathsf{x}_{\bar{\mathsf{t}},f,2s_i}^{(i)}\\
  \quad
  -\sum_{h=0}^{2s_j-1} \beta^h\ \frac{\sinh(\xi_i^{(2s_i)}-\zeta)}{\sinh(\xi_j^{(h)}-\zeta)}\,
  \underset{(m,k)\not=(j,h)\ }{\prod_{m=1}^\mathsf{N}\!\prod_{k=0}^{2s_m-1}}
  \frac{\sinh(\xi_i^{(2s_i)}-\xi_m^{(k)})}{\sinh(\xi_j^{(h)}-\xi_m^{(k)})}\ \mathsf{x}_{\bar{\mathsf{t}},f,h}^{(j)}.
\end{multline}
Hence, for all values of $e^\alpha$ for which there exists $\zeta_0$  (distinct from $\zeta_1,\ldots,\zeta_{\mathsf{N}_s}$ modulo $i\pi$) such that $\det_\mathsf{N}\left[\mathcal{M}_{\bar{\mathsf{t}}}(e^\alpha,\zeta_0)\right]\not=0$, the solution to the system \eqref{sys-0-2} is unique if we fix the value of $q_0\equiv Q(\zeta_0)\in\mathbb{C}$.
In other words it means that, in this case, there exists (up to a global normalization related to the value of $Q(\zeta_0)$) one and only one solution $Q(\lambda)\in\mathbb{T}_{\mathsf{N}_s}[\lambda]$ to the functional equation \eqref{eq-inh}.
It is easy to see, from the explicit expression \eqref{mat-sys0}, that $\det_\mathsf{N}\left[\mathcal{M}_{\bar{\mathsf{t}}}(\beta,\zeta)\right]$ is a non-zero polynomial in $\beta$ of degree exactly\footnote{This comes from the fact that $\mathsf{q}_{\,\bar{\mathsf{t}},2s_i}^{(i)}\not=0$, $\forall i\in\{1,\ldots,\mathsf{N}\}$ (see Remark~\ref{rem-eigenvector-D}).} $\mathsf{N}_s$ (hence non-constant), which can therefore be written in the following form in terms of its $\mathsf{N}_s$ zeros $\beta_{\bar{\mathsf{t}},j}(\zeta)$, $1\le j\le \mathsf{N}_s$:
\begin{equation}\label{det-M}
  \det_\mathsf{N}\left[\mathcal{M}_{\bar{\mathsf{t}}}(\beta,\zeta)\right]
    =C_{\bar{\mathsf{t}}} \prod_{j=1}^{\mathsf{N}_s} (\beta-\beta_{\bar{\mathsf{t}},j}(\zeta)),
  \qquad\text{with}
  \quad
  C_{\bar{\mathsf{t}}} = \prod_{j=1}^{\mathsf{N}_s} \mathsf{x}_{\bar{\mathsf{t}},2s_j}^{(j)}.
\end{equation}
Note that we have in particular
\begin{align}
  \det_\mathsf{N}\left[\mathcal{M}_{\bar{\mathsf{t}}}(0,\zeta_0)\right]
  &=(-1)^{\mathsf{N}_s}\, C_{\bar{\mathsf{t}}}\, \prod_{j=1}^{\mathsf{N}_s} \beta_{\bar{\mathsf{t}},j}(\zeta_0) 
\label{det-coeff+} \\
  &=
\frac{\prod_{a<b}^{\mathsf{N}}\sinh(\xi_a^{(2s_a)}-\xi_b^{(2s_b)})\,\sinh (\xi_b^{(0)}-\xi_a^{(0)})}
       {\prod_{a,b=1}^{\mathsf{N}}\sinh (\xi_a^{(2s_a)}-\xi_b^{(0)})}
\prod_{n=1}^{\mathsf{N}}
\frac{\prod_{k=0}^{\mathsf{N}_s}\sinh (\xi_n^{(2s_n)}-\zeta_k)}
       {\prod_{\substack{ k=0\\ k\neq j_{n,0}}}^{\mathsf{N}_s}\sinh (\xi_n^{(0)}-\zeta_k)}.\notag
\end{align}
Hence, for a given value of $\zeta_0$, there exists at most $\mathsf{N}_s$ values of $e^\alpha$ such that $\det_\mathsf{N}\left[\mathcal{M}_{\bar{\mathsf{t}}}(e^\alpha,\zeta_0)\right]=0$.
So, if any, there exists at most finitely many values of $e^\alpha$ such that, for any $\zeta_0$  distinct from $\zeta_1,\ldots,\zeta_{\mathsf{N}_s}$ modulo $i\pi$, $\det_\mathsf{N}\left[\mathcal{M}_{\bar{\mathsf{t}}}(e^\alpha,\zeta_0)\right]=0$. Note that this last condition is equivalent to the fact that the functional equation \eqref{eq-inh} admits at least two independent solutions\footnote{Indeed, the fact that $\det_\mathsf{N}\left[\mathcal{M}_{\bar{\mathsf{t}}}(e^\alpha,\zeta_0)\right]=0$ for infinitely many distinct values of $\zeta_0$ modulo $i\pi$ means that, for each of these values of $\zeta_0$, one can construct a non-zero solution to the system \eqref{sys-0-2} with $q_0=0$, i.e. a non-zero solution $Q(\lambda)\in\mathbb{T}_{\mathsf{N}_s}[\lambda]$ to the functional equation \eqref{eq-inh} such that $Q(\zeta_0)=0$. It is easy to see that this means that the space of solutions is at least of dimension 2.}.

Let us finally suppose that, for a solution $Q(\lambda)$ to the functional equation \eqref{eq-inh}, there exists $n\in\{1,\ldots,\mathsf{N}\}$ such that $Q(\xi_n^{(0)})=0$. This implies that, $\forall h\in\{0,\ldots,2s_n\}$, $Q(\xi_n^{(h)})=e^{h\alpha}\, \mathsf{x}_{\bar{\mathsf{t}},f,h}^{(n)}\, Q(\xi_n^{(0)})=0$, i.e. that $Q(\lambda)$ is of the form
\begin{equation}
   Q(\lambda)=\prod_{h=0}^{2s_n}\sinh(\lambda-\xi_n^{(h)})\ Q_n(\lambda),
\end{equation}
where ${Q}_n(\lambda)\in\hat{\mathbb{T}}_{\mathsf{N}_n}[\lambda]$ with $\mathsf{N}_n\equiv\mathsf{N}_s-(2s_n+1)=2\sum_{ j\not= n}s_j-1$. Then the functional equation \eqref{eq-inh} for $Q(\lambda)$ can be simplified into a functional equation for $Q_n(\lambda)$ of the form:
\begin{equation}\label{eq-Qtilde}
  \mathsf{\bar{t}}(\lambda)\, Q_n(\lambda)
      =-e^{-\alpha}\, f_n(\lambda)\,a(\lambda)\, Q_n(\lambda-\eta)
      +\frac{e^{\alpha}\, d(\lambda)}{f_n(\lambda+\eta)}\, Q_n(\lambda+\eta)
      +F_{\alpha+\bar{\lambda}_n,n}(\lambda),
\end{equation}
where
\begin{align}\label{fn}
  &f_n(\lambda)=e^\lambda\,\frac{\sinh(\lambda-\xi_n-(s_n+1)\eta)}{\sinh(\lambda-\xi_n+s_n\eta)},\\
  &F_{x,n}(\lambda)=2 e^{-\frac{\mathsf{N}_s+1}{2}\eta}
     \sinh\!\Bigg(\!
     \lambda-x-\xi_n^{(0)} +\sum_{\substack{j=1\\ j\not=n}}^\mathsf{N}\sum_{h_j=1}^{2s_j}\!\xi_j^{(h_j)}+\frac{\mathsf{N}_s+1}{2}\eta\!
     \Bigg) 
     \prod_{\substack{j=1\\ j\not= n}}^\mathsf{N}\prod_{h_j=0}^{2s_j}\!\sinh\!\big(\lambda-\xi_j^{(h_j)}\big),\notag
\end{align}
and where $\bar{\lambda}_n$ stands for the sum of the roots of the trigonometric polynomial $Q_n(\lambda)$.
Note that the apparent pole in $f_n(\lambda)$ or $f_n(\lambda+\eta)^{-1}$ is cancelled by the corresponding factor in $a(\lambda)$ or $d(\lambda)$, and that $F_{\alpha+\bar{\lambda}_n,n}(\lambda)$ cancels the leading behavior when $\lambda\to\pm\infty$, proportional to $e^{\pm(\mathsf{N}+\mathsf{N}_n+1)\lambda}$, of the other two terms of r.h.s. of \eqref{eq-Qtilde}.
In other words, it means that \eqref{eq-Qtilde} is an identity between trigonometric polynomials in $\mathbb{T}_{\mathsf{N}+\mathsf{N}_n-1}[\lambda]$.
It therefore remains to see whether or not such an equation admits a non-zero solution ${Q}_n$.

We can use a similar reasoning as previously. As an identity in $\mathbb{T}_{\mathsf{N}+\mathsf{N}_n-1}[\lambda]$, the equation \eqref{eq-Qtilde} can be rewritten as a system of $\mathsf{N}+\mathsf{N}_n$ equations obtained by particularizing \eqref{eq-Qtilde} at the $\mathsf{N}+\mathsf{N}_n$ points $\xi_j^{(h_j)}$, $j\in\{1,\ldots,\mathsf{N}\}\setminus\{n\}$, $h_j\in\{0,\ldots,2s_j\}$, i.e. by the union of the following $\mathsf{N}-1$ subsystems for $j\in \{1,\ldots,\mathsf{N}\}\setminus\{n\}$:
\begin{equation}\label{sub-sys-tilde}
 \big(X_\alpha\, X_{f_n}^{(j)}\big)^{-1}\, D_{\bar{\mathsf{t}}}^{(j)}\, 
   \big(X_{\alpha}\, X_{f_n}^{(j)}\big)
 \cdot
  \left( 
\begin{array}{c}
{Q}_n(\xi_j^{(0)}) \\ 
{Q}_n(\xi_j^{(1)}) \\ 
\vdots \\ 
{Q}_n(\xi_j^{(2s_{j})})
\end{array}
\right) 
=\left( 
\begin{array}{c}
0 \\ 
\vdots \\ 
\vdots \\ 
0
\end{array}
\right),
\end{equation}
where $X_\alpha$ and $X_{f_n}^{(j)}$ are defined as in  \eqref{Xn}.
Each of these subsystems admits, up to an overall normalization coefficient ${Q}_n(\xi_j^{(0)})$, only one non-zero solution which is given in terms of the vector $\mathbf{q}_{\,\bar{\mathsf{t}}}^{(n)}$ \eqref{q-vect}-\eqref{rec-y}.
On the other hand, every polynomial $Q_n(\lambda)\in\mathbb{T}_{\mathsf{N}_n}[\lambda]$ is completely characterized by the data of its values at $\mathsf{N}_n+1$ different points, for instance at the points $\xi_j^{(h)}$, $j\in\{1,\ldots\mathsf{N}\}\setminus\{n\}$, $h\in\{0,1,\ldots,2s_j-1\}$.
Hence the above system of $\mathsf{N}+\mathsf{N}_n$ equations is equivalent, for $Q_n(\lambda)\in\mathbb{T}_{\mathsf{N}_n}[\lambda]$, to the union of the following coupled sub-systems for $j\in\{1,\ldots\mathsf{N}\}\setminus\{n\}$:
\begin{equation}\label{sys-bis}
  \left\{
  \begin{split}
   &Q_n(\xi_j^{(h)})=e^{h\alpha}\ \mathsf{x}_{\bar{\mathsf{t}},f_n,h}^{(j)}\ Q_n(\xi_j^{(0)}),
     \qquad 1\le h\le 2s_j,\\
   & Q_n(\xi_j^{(2s_j)})
     =\sum_{\substack{m=1 \\ m\not=n}}^{\mathsf{N}}\sum_{h=0}^{2s_j}\ 
     \underset{k\not=n,\, (k,\ell)\not=(m,h) }{\prod_{k=1}^\mathsf{N} \prod_{\ell=0}^{2s_k-1}}
     \hspace{-3mm}
     \frac{\sinh(\xi_j^{(2s_j)}-\xi_k^{(\ell)})}{\sinh(\xi_m^{(h)}-\xi_k^{(\ell)})}\
     Q_n(\xi_m^{(h)}).
  \end{split}
  \right.
\end{equation}
This system admits a non-zero solution if and only if the following  homogeneous system of $\mathsf{N}-1$ equations in the $\mathsf{N}-1$ variables $Q_n(\xi_j^{(0)})$,  $j\in\{1,\ldots,\mathsf{N}\}\setminus\{n\}$ admits a non-zero solution:
\begin{equation}
  \sum_{\substack{j=1\\ j\not=n}}^\mathsf{N} \left[\mathcal{M}_{\bar{\mathsf{t}},n}(e^\alpha)\right]_{i,j}\, Q_n(\xi_j^{(0)})
   = 0,
   \qquad i\in\{1,\ldots\mathsf{N}\}\setminus\{n\},
\end{equation}
where the matrix $\mathcal{M}_{\bar{\mathsf{t}},n}(\beta)$ is defined as
\begin{equation}\label{mat-sys-bis}
  \left[\mathcal{M}_{\bar{\mathsf{t}},n}(\beta)\right]_{i,j}
  =\delta_{i,j}\  \beta^{2s_i}\, \mathsf{x}_{\bar{\mathsf{t}},f_n,2s_i}^{(i)}
  -\sum_{h=0}^{2s_j-1} \beta^h \hspace{-1mm}
  \underset{m\not=n,\, (m,k)\not=(j,h)\ }{\prod_{m=1}^\mathsf{N}\prod_{k=1}^{2s_m-1}}\hspace{-3mm}
  \frac{\sinh(\xi_i^{(2s_i)}-\xi_m^{(k)})}{\sinh(\xi_j^{(h)}-\xi_m^{(k)})}\ \mathsf{x}_{\bar{\mathsf{t}},f_n,h}^{(j)},
\end{equation}
for $i,j\in\{1,\ldots,\mathsf{N}\}\setminus\{n\}$.
Since the determinant of $\mathcal{M}_{\bar{\mathsf{t}},n}(\beta)$ is a non-zero polynomial in $\beta$ of degree  $\mathsf{N}_n+1$, it vanishes for at most $\mathsf{N}_n+1$ different values of $\beta$, which ends the proof of the theorem.
\qed

\section{Reformulation of the SOV spectrum in terms of solutions of a homogeneous functional equation}
\label{sec-hom}

We have seen in the previous section that the eigenvalues and the eigenvectors of the ($\kappa$-twisted) antiperiodic transfer matrix could be completely characterized by trigonometric polynomial solutions (i.e. solutions of the form \eqref{Q-2pi}) of a certain degree $\mathsf{N}_s$, of some particular inhomogeneous second-order finite-difference functional equation. This correspondence is bijective, i.e. with each admissible solution of the associated system  of Bethe equations correspond exactly one eigenvalue and one eigenvector (given respectively by \eqref{bethe-inh} and by \eqref{eigenl-bis}-\eqref{eigenr-bis} with \eqref{Q-t}). This follows from the SOV characterization of the eigenvalues \eqref{t-form}-\eqref{discrete-sys} and of the eigenvectors \eqref{eigenl}-\eqref{eigenr} by analytical and very simple arguments only.

Nevertheless, the resolution of a system of Bethe equations with such an additional inhomogeneous contribution is likely to involve extra difficulties compared to what happens for a usual (homogeneous) system of Bethe equations.
Hence it is surely preferable, whenever it is possible and as long as the functions involved are not too complicated, to express our eigenfunctions and eigenvectors in terms of solutions of {\em homogeneous} Bethe equations. In the case of antiperiodic boundary conditions, and unlike what happens for periodic boundary conditions, it is well known that the homogeneous second-order finite-difference equation which is the continuous analog of \eqref{eq-left},
\begin{equation}\label{Bax-hom-eq}
   \bar{\mathsf{t}}(\lambda)\, Q(\lambda )
   =-a(\lambda )\, Q(\lambda -\eta )+d(\lambda)\, Q(\lambda +\eta ),
\end{equation}
does not admit, for $\bar{\mathsf{t}}(\lambda)\in\Sigma_{\bar{\mathsf{T}}}$, any solution of the type \eqref{Q-2pi}, since the $i\pi$-quasi-periodicity behavior 
of the left and right hand sides of \eqref{Bax-hom-eq} would not coincide. 
Hence the whole problem is to identify the functional form of a complete family of $Q$-solutions to \eqref{Bax-hom-eq} for all $\bar{\mathsf{t}}(\lambda)\in\Sigma_{\bar{\mathsf{T}}}$, so as to be able to write explicitly a complete system of homogenous Bethe equations for these eigenvalues. Depending on the model we consider, this may be a difficult problem\footnote{In particular, for XXZ chains with completely general boundary conditions, this still remains an open problem.} which is usually solved by means of the construction of Baxter's $Q$-operator \cite{Bax72,Bax82L}. Note that, even for periodic models, such a construction may be quite involved, and has already generated an impressive quantity of works (see for instance \cite{PasG92,BazLZ97,AntF97,BazLZ99,Der99,Pro00,Kor06,Der08,BazLMS10,Man14}).

A construction of the $Q$-operator, in the case of the antiperiodic integrable XXZ chain of spin $s$, has been proposed in \cite{YunB95} (see \cite{BatBOY95} for the spin-1/2 case) using the ``pair-propagation through a vertex" property of the six-vertex model.
The eigenvalues $Q(\lambda)$ of this operator have there been found to be of the form\footnote{Note that, from a purely functional viewpoint and without knowing anything about the construction of \cite{BatBOY95,YunB95}, it is not unreasonable to postulate that \eqref{Bax-hom-eq} admits $Q$-solutions with a double period compared to the $Q$-solutions of \eqref{hom-TQ} in the periodic case: indeed, the functional equation itself (once conveniently renormalized), has coefficient with double periodicity with respect to its periodic counterpart. The degree $\mathsf{N}_s$ of the solution is then completely fixed by the leading asymptotic behavior of the two members of \eqref{Bax-hom-eq}.}
\begin{equation}\label{Q-homo-Ansatz}
Q(\lambda )=\prod_{a=1}^{\mathsf{N}_s}\sinh \left( \frac{\lambda -\lambda_a}{2}\right) ,
\qquad
\mathsf{N}_s\equiv 2\sum_{a=1}^{\mathsf{N}} s_{a},
\end{equation}
in terms of some roots $\lambda_1,\ldots,\lambda_{\mathsf{N}_s}$.
By means of such solutions\footnote{The construction of \cite{BatBOY95,YunB95} has actually been achieved for the homogeneous model, a case for which Theorem~\ref{th-eigen-t} does not hold. The method used in \cite{BatBOY95,YunB95} should nevertheless be generalizable without problem to the inhomogeneous model considered in the present paper.}, the spectral problem for the antiperiodic transfer matrix can then be reformulated in terms of a system of Bethe equations for the roots $\lambda_1,\ldots,\lambda_{\mathsf{N}_s}$ of the functions \eqref{Q-homo-Ansatz}. 
The explicit construction of the $Q$-operator ensures in principle, together with the SOV characterization of the transfer matrix spectrum and eigenstates provided by Theorem~\ref{th-eigen-t}, the completeness\footnote{modulo the problem of admissibility of the corresponding solutions, see Remarks~\ref{rem-xi_j} and \ref{rem-Q+ipi}.} of the solution issued from \eqref{Bax-hom-eq}-\eqref{Q-homo-Ansatz}.
However we would like to recall that the construction of \cite{BatBOY95,YunB95} relies in fact on a conjecture (the assumption that the operator called in \cite{BatBOY95,YunB95} $Q_R(v)$ is invertible at some point $v=v_0$) so that, except if we can prove the aforementioned completeness property of the corresponding family of solutions by independent arguments, the latter is conditioned to the validity of this conjecture.
This is precisely the purpose of the present section to provide such an independent proof of this property:
we explain here how to reformulate the SOV characterization of the spectrum of Theorem~\ref{th-eigen-t} in terms of solutions \eqref{Q-homo-Ansatz} to \eqref{Bax-hom-eq} using only very simple analytical arguments, in the spirit of what has been done in the previous section.
As in the previous section, we show that this reformulation is an equivalence, so that it preserves the completeness property of Theorem~\ref{th-eigen-t}. 
Once again, we insist on the fact that our proof  uses only simple analytical properties of the $Q$-function, so that we hope that our study can be extended to other models solvable by SOV, and in particular to cases in which the explicit construction of the $Q$-operator is presently not known.

Our result can be stated as follows.
 
 \begin{theorem}\label{th-hom}
 Under the condition \eqref{cond-inh}, the following propositions are equivalent:
\begin{enumerate}
   \item\label{cond1-hom}
   $\mathsf{\bar{t}}(\lambda)$ is an eigenvalue function of the $\kappa$-twisted antiperiodic transfer matrix $\bar{\mathsf{T}}^{(\kappa)}(\lambda)$ for any $\kappa\in\mathbb{C}\setminus\{0\}$
   (i.e. $\mathsf{\bar{t}}(\lambda)\in\Sigma_{\bar{\mathsf{T}}}$);
   \item\label{cond2-hom}
   $\mathsf{\bar{t}}(\lambda)$ is an entire function of $\lambda$ such that $\mathsf{\bar{t}}(\lambda+i\pi)=(-1)^{\mathsf{N}-1}\,\mathsf{\bar{t}}(\lambda)$, and there exists a unique function $Q(\lambda)$ of the form \eqref{Q-homo-Ansatz} such that $\mathsf{\bar{t}}(\lambda)$ and $Q(\lambda)$ satisfy the homogeneous functional equation \eqref{Bax-hom-eq}.
This function $Q(\lambda)$ is  such that  $\big( Q(\xi_j^{(0)}), Q(\xi_j^{(0)}+i\pi)\big)\not=(0,0)$ for all $j\in\{1,\ldots,\mathsf{N}\}$. 
   \item\label{cond3-hom}
   $\mathsf{\bar{t}}(\lambda)$ is an entire function of $\lambda$, and there exist $\epsilon \in\{-1,+1\}$ and a function $Q(\lambda)$ of the form \eqref{Q-homo-Ansatz} satisfying the relation
\begin{equation}
Q(\lambda )\,Q(\lambda +i\pi -\eta )+Q(\lambda +i\pi )\,Q(\lambda -\eta
)=d(\lambda )\, w_{\epsilon}(\lambda ),  \label{W-relation}
\end{equation}
such that
\begin{equation}
\mathsf{\bar{t}}(\lambda ) 
=\frac{Q(\lambda +\eta )\,Q(\lambda +i\pi -\eta )-Q(\lambda +\eta +i\pi
)\,Q(\lambda -\eta )}{w_{\epsilon}(\lambda )},  \label{expr-t-Q}
\end{equation}
with
\begin{equation}
w_\epsilon(\lambda )
=2\epsilon \left(\frac{ i}{2}\right)^{\!\mathsf{N}_s}  \, 
\prod_{n=1}^\mathsf{N}\prod_{h_n=1}^{2s_n-1}\sinh(\lambda-\xi_n^{(h_n)}).
\label{w-eps}
\end{equation}
The couple $(\epsilon,Q)$ is unique and the function $Q(\lambda)$ is  such that  $\big( Q(\xi_j^{(0)}), Q(\xi_j^{(0)}+i\pi)\big)\not=(0,0)$ for all $j\in\{1,\ldots,\mathsf{N}\}$.
\end{enumerate}
 \end{theorem}
 
 Before proving this theorem, let us briefly comment about its content.
 
 This theorem provides in fact two equivalent characterizations of the transfer matrix spectrum $\Sigma_{\bar{\mathsf{T}}}$ in terms of functions  of the form \eqref{Q-homo-Ansatz}. The first one (item  {\it \ref{cond2-hom}.}) corresponds to the usual characterization in terms of the solutions of the homogeneous Baxter functional equation \eqref{Bax-hom-eq}. The statement  is that the description in terms of solutions of the form \eqref{Q-homo-Ansatz} to \eqref{Bax-hom-eq} is indeed complete. We can therefore use the solutions of the corresponding Bethe equations (see Corollary~\ref{cor-Bethe-hom}) to analyze the spectrum and construct the corresponding eigenvectors. Note that the proof of  {\it \ref{cond2-hom}.} provides an implicit construction of the $Q$-operator, as a diagonal matrix in the eigenbasis of the ($\kappa$-twisted) antiperiodic transfer matrix.
 
 The second characterization (item  {\it \ref{cond3-hom}.}) is maybe less usual. It proposes an alternative way to obtain the Bethe roots without solving explicitly the Bethe equations, but instead based on the identification of the $Q$-solutions to the relation \eqref{W-relation}. 
This characterization is related to the fact that, as a second order difference equation in the function $Q(\lambda)$, \eqref{Bax-hom-eq} is actually expected to admit, for each fixed function $\bar{\mathsf{t}}(\lambda)$,  two linearly independent $Q$-solutions \cite{BazLZ97,KriLWZ97}. It is indeed easy to see that, if $Q(\lambda)$ satisfies the homogeneous functional equation \eqref{Bax-hom-eq} for a given function $\bar{\mathsf{t}}(\lambda)$ such that $\mathsf{\bar{t}}(\lambda+i\pi)=(-1)^{\mathsf{N}-1}\,\mathsf{\bar{t}}(\lambda)$, then the function $Q(\lambda+i\pi)$ satisfies the $i\pi$-twisted homogeneous functional equation
\begin{equation}\label{hom-eq-bis}
   \bar{\mathsf{t}}(\lambda)\, Q(\lambda+i\pi )
   =a(\lambda )\, Q(\lambda+i\pi -\eta )-d(\lambda)\, Q(\lambda +i\pi+\eta ),
\end{equation}
or, in other words, the function $\widetilde{Q}(\lambda)\equiv e^{i\pi\frac{\lambda}{\eta}} Q(\lambda+i\pi)$ provides another solution to the homogeneous functional equation \eqref{Bax-hom-eq} associated to $\bar{\mathsf{t}}(\lambda)$ (of course, if $Q(\lambda)$ is of the form \eqref{Q-homo-Ansatz}, this is not the case for $\widetilde{Q}(\lambda)$). The quantum Wronskian of these two solutions is then proportional (with proportionality factor $e^{i\pi\frac{\lambda}{\eta}}$) to the following function:
\begin{equation}\label{def-W}
W_Q(\lambda )\equiv Q(\lambda +i\pi )\, Q(\lambda -\eta )+Q(\lambda )\, Q(\lambda+i\pi -\eta ).
\end{equation}
It can be shown that the latter has to satisfy the identity \eqref{W-relation}, a condition which is in fact equivalent to the $i\pi$-quasi-periodicity property of $\bar{\mathsf{t}}(\lambda)$ (see Lemma~\ref{lem-WQ}). Hence, as functions, these two solutions $Q(\lambda)$ and $\widetilde{Q}(\lambda)$ are independent and they can be used to construct the function $\bar{\mathsf{t}}(\lambda)$ as in \eqref{expr-t-Q}.
Item  {\it \ref{cond3-hom}.} of Theorem~\ref{th-hom} states that, conversely, solutions of the form \eqref{Q-homo-Ansatz} to the Wronskian identity~\eqref{W-relation} completely characterize the transfer matrix spectrum $\Sigma_{\bar{\mathsf{T}}}$.

\begin{rem}\label{rem-spin1/2}
The formulation of item {\it \ref{cond3-hom}.} of Theorem~\ref{th-hom} can be simplified in the spin-1/2 case: in that case the function $w_\epsilon(\lambda)$ \eqref{w-eps} is in fact a constant, so that the function $\bar{\mathsf{t}}(\lambda)$ defined through \eqref{expr-t-Q} is automatically entire. 
\end{rem}

\begin{rem}
The Wronskian identity \eqref{W-relation} notably implies the following sum-rule for the roots $\lambda_j$ of the function $Q(\lambda)$, which can be obtained from the comparison of the $\lambda\to\pm\infty$ leading asymptotic behavior of the two members of \eqref{W-relation}:
\begin{equation}
\exists m\in \mathbb{Z},\quad
\sum_{a=1}^{\mathsf{N}_s}\lambda _{a}=i\Big(\frac{1-\epsilon}{2}+2m\Big)\pi +
\sum_{b=1}^{\mathsf{N}}\sum_{h_b=0}^{2s_b-1}\xi _{b}^{(h_b)}
-\frac{\mathsf{N}_s}{2}\eta.  \label{asymp-cond}
\end{equation}
\end{rem}

\begin{rem}\label{rem-prop-vec}
For $\bar{\mathsf{t}}(\lambda)\in\Sigma_{\bar{\mathsf{T}}}$, the two vectors $\mathbf{Q}^{(n)}$ and $\widetilde{\mathbf{Q}}^{(n)}$ of size $2s_n$ defined from the two independent solutions $Q(\lambda)$ and $\widetilde{Q}(\lambda)\equiv e^{i\pi\frac{\lambda}{\eta}} Q(\lambda+i\pi)$ to \eqref{Bax-hom-eq} as
\begin{equation}\label{def-Q-tildeQ}
  \mathbf{Q}^{(n)}=\begin{pmatrix}
    Q(\xi_n^{(0)})  \\ \vdots \\ Q(\xi_n^{(h_n)}) \\ \vdots \\ Q(\xi_n^{(2s_n)})
    \end{pmatrix}
    \qquad
    \text{and}
    \qquad
   \widetilde{\mathbf{Q}}^{(n)}=\begin{pmatrix}
    Q(\xi_n^{(0)}+i\pi) \\ \vdots \\ (-1)^{h_n}\,Q(\xi_n^{(h_n)}+i\pi) \\ \vdots \\ (-1)^{2s_n}\,Q(\xi_n^{(2s_n)}+i\pi) 
    \end{pmatrix}
\end{equation}
should instead be proportional, since they both satisfy \eqref{eq-left} and therefore both belong to the one-dimensional nullspace of the matrix $D_{\bar{\mathsf{t}}}^{(n)}$ (note that this is in agreement with the form \eqref{W-relation} of the Wronskian). 
One can show conversely (see Proposition~\ref{prop-WQ}) that the proportionality of these two vectors (together with the fact that at least one of them is non-zero) for each $n\in\{1,\ldots,\mathsf{N}\}$ is a sufficient condition for a function $\mathsf{\bar{t}}(\lambda)$ defined by \eqref{expr-t-Q} to belong to $\Sigma_{\bar{\mathsf{T}}}$. Here there is no need to impose in addition that the function \eqref{expr-t-Q} is entire or satisfies the adequate $i\pi$-quasi-periodicity, since these properties are simple consequences of the form \eqref{Q-homo-Ansatz} of $Q(\lambda)$ and of the proportionality property of the vectors \eqref{def-Q-tildeQ}.
Hence, solving this proportionality condition for functions of the form \eqref{Q-homo-Ansatz} is an alternative to solving the Bethe equations and leads to the same sets of Bethe roots.
Note that, in the spin-$1/2$ case, for $Q(\lambda)$ of the form \eqref{Q-homo-Ansatz}, the proportionality condition of the two vectors \eqref{def-Q-tildeQ} is strictly equivalent to the Wronskian identity \eqref{W-relation}, as it is clear from the proof of Proposition~\ref{prop-WQ}.
\end{rem}

 \begin{rem}\label{rem-Q+ipi}
The fact that we have these two solutions also means that we can construct the eigenvectors \eqref{eigenl} and \eqref{eigenr} either from $Q(\lambda)$ or from $Q(\lambda+i\pi)$, i.e. from either of the two vectors \eqref{def-Q-tildeQ}. This enables us to relax our notion of {\em admissible} solution with respect to Remark~\ref{rem-xi_j}: an admissible solution $Q(\lambda)$ of Equation~\eqref{Bax-hom-eq} is here defined to be a solution such that $\big( Q(\xi_j^{(0)}), Q(\xi_j^{(0)}+i\pi)\big)\not=(0,0)$ for all $j\in\{1,\ldots,\mathsf{N}\}$. This condition is indeed sufficient to ensure that at least one of the two vectors \eqref{def-Q-tildeQ} is non-zero, and as a consequence that the corresponding function $\mathsf{\bar{t}}(\lambda)$ belongs to the transfer matrix spectrum.
\end{rem}

As previously, item  {\it \ref{cond2-hom}.} of this theorem enables us to completely characterize the spectrum and eigenstates of the ($\kappa$-twisted) antiperiodic transfer matrix in terms of the admissible solutions  $\Lambda=\{\lambda_1,\ldots,\lambda_{\mathsf{N}_s}\}$ of the system of homogeneous Bethe equations following from the cancellation conditions of all apparent poles at $\lambda=\lambda_j$, $1\le j\le \mathsf{N}_s$,  of the function
\begin{equation}\label{bethe-hom}
 \bar{\mathsf{t}}_{\Lambda}(\lambda)\equiv
-a(\lambda)\prod_{j=1}^{\mathsf{N}_s}\frac{\sinh\big(\frac{\lambda-\lambda_j-\eta}{2}\big)}{\sinh\big(\frac{\lambda-\lambda_j}{2}\big)}
   +d(\lambda)\prod_{j=1}^{\mathsf{N}_s}\frac{\sinh\big(\frac{\lambda-\lambda_j+\eta}{2}\big)}{\sinh\big(\frac{\lambda-\lambda_j}{2}\big)}.
\end{equation}

\begin{corollary}\label{cor-Bethe-hom}
Under the condition \eqref{cond-inh}, there exists a one-to-one correspondence between $\Sigma_{\bar{\mathsf{T}}}$ and the set $\Sigma_{\mathrm{BAE}}$ of solutions $\Lambda=\{\lambda_1,\ldots,\lambda_{\mathsf{N}_s}\}$ to the system of $\mathsf{N}_s$ homogeneous Bethe equations following from the condition that the function \eqref{bethe-hom} is entire 
and which satisfy in addition the following conditions:
\begin{enumerate} 
  \item[\sf{({\makebox[0.6em]{i}})}] $\forall j\in\{1,\ldots,\mathsf{N}_s\}, \quad\Im\lambda_j\in[0,2\pi[\ $;
  \item[\sf{({\makebox[0.6em]{ii}})}]
  $\forall k\in\{1,\ldots,\mathsf{N}\}, \quad
  \big\{\xi_k^{(0)},\xi_k^{(0)}+i\pi\big\} \not\subset
  \big\{\lambda_j+2 i\ell\pi \mid 1\le j\le \mathsf{N}_s,\ \ell\in\mathbb{Z} \big\}\,$;
  \item[\sf{({\makebox[0.6em]{iii}})}]
  the function \eqref{bethe-hom} satisfies the quasi-periodicity property $\bar{\mathsf{t}}_{\Lambda}(\lambda+i\pi)=(-1)^{\mathsf{N}-1}\bar{\mathsf{t}}_{\Lambda}(\lambda)$.
\end{enumerate}
The eigenvalue $\bar{\mathsf{t}}_{\Lambda}(\lambda)\in\Sigma_{\bar{\mathsf{T}}}$ associated to $\Lambda\in\Sigma_{\mathrm{BAE}}$ is then given by the entire function \eqref{bethe-hom}.
For any $\kappa\in\mathbb{C}\setminus\{0\}$, the one-dimensional left and right eigenspaces of $\bar{\mathsf{T}}^{(\kappa)}(\lambda)$ associated with the eigenvalue $\bar{\mathsf{t}}_{\Lambda}(\lambda)\in\Sigma_{\bar{\mathsf{T}}}$ are respectively spanned by the vectors
\begin{equation}\label{l-space}
     \sum_{\mathbf{h}}
     \prod_{j=1}^\mathsf{N} \left\{(\epsilon \kappa)^{h_j}\, Q^{(\epsilon)}_{\Lambda}(\xi_j^{(h_j)})\right\}\,
     \prod_{i<j}\sinh(\xi_j^{(h_j)}-\xi_i^{(h_i)})\, \bra{\mathbf{h} }  ,\ \ \epsilon=\pm1,
\end{equation}
and
\begin{equation}\label{r-space}
     \sum_{\mathbf{h}}
     \prod_{j=1}^\mathsf{N} \Bigg\{\! (-\epsilon\kappa)^{-h_j}
     \Bigg(\!\prod_{\ell=0}^{h_j-1}\!\frac{a(\xi_n^{(\ell)})}{d(\xi_n^{(\ell+1)})}\!\Bigg)
     Q^{(\epsilon)}_{\Lambda}(\xi_j^{(h_j)})\!\Bigg\}
     \prod_{i<j}\sinh(\xi_j^{(h_j)}-\xi_i^{(h_i)})\, \ket{\mathbf{h} }   ,\ \ \epsilon=\pm1.
\end{equation}
Here the two functions $Q^{(\epsilon)}_{\Lambda}(\lambda)$, $\epsilon\in\{+1,-1\}$, are defined in terms of the set of Bethe roots $\Lambda$ as follows:
\begin{equation}\label{Q-t-pm}
   Q_{\Lambda}^{(\epsilon)}(\lambda)=
   \prod_{j=1}^{\mathsf{N}_s}\sinh\left(\frac{\lambda-\lambda_j}{2}+\frac{(1-\epsilon)}{4} i\pi\right).
\end{equation}
\end{corollary}

We shall now prove Theorem~\ref{th-hom}, for which we first formulate a couple of useful lemmas.

\begin{lemma}\label{lem-uniqueness}
Let $Q_1(\lambda)$ and $Q_2(\lambda)$ be of the form \eqref{Q-homo-Ansatz} such that
the two functions
\begin{equation}\label{t-egal}
   \frac{-a(\lambda )\, Q_1(\lambda -\eta )+d(\lambda)\, Q_1(\lambda +\eta )}{Q_1(\lambda )}
   \quad
   \text{and}
   \quad
    \frac{-a(\lambda )\, Q_2(\lambda -\eta )+d(\lambda)\, Q_2(\lambda +\eta )}{Q_2(\lambda )}
\end{equation}
are identically equal. Then $Q_1(\lambda)=Q_2(\lambda)$ identically.
\end{lemma}

\begin{proof}
  The fact that the two functions \eqref{t-egal} coincide implies that
  \begin{equation}\label{W12-eq}
   a(\lambda)\, W_{12}(\lambda) = -d(\lambda)\, W_{12}(\lambda+\eta),
  \end{equation}
  where
  \begin{equation}
    W_{12}(\lambda)\equiv Q_1(\lambda-\eta)\, Q_2(\lambda)-Q_1(\lambda)\, Q_2(\lambda-\eta).
  \end{equation}
  It is clear that $W_{12}(\lambda)\in\mathbb{T}_{2\mathsf{N}_s-2}[\frac{\lambda}{2}]$, so that the only solution to \eqref{W12-eq} is $W_{12}(\lambda)=0$, which means that $Q_1(\lambda)=Q_2(\lambda)$.
\end{proof} 
 
\begin{lemma}\label{lem-WQ}
Let $Q(\lambda)$ be of the form \eqref{Q-homo-Ansatz}. Then the function $\bar{\mathsf{t}}(\lambda)$ defined as 
\begin{equation}\label{def-bart}
  \bar{\mathsf{t}}(\lambda)\equiv 
  \frac{-a(\lambda )\, Q(\lambda -\eta )+d(\lambda)\, Q(\lambda +\eta )}{Q(\lambda )},
\end{equation}
satisfies the quasi-periodicity property $\bar{\mathsf{t}}(\lambda +i\pi )=(-1)^{\mathsf{N}-1}\,\bar{\mathsf{t}}(\lambda )$ if and only if $Q(\lambda)$ satisfies the Wronskian identity \eqref{W-relation} for some $\epsilon\in\{-1,1\}$.
\end{lemma}

\begin{proof}
Taking into consideration the functional form of the coefficients $a(\lambda )$ and $d(\lambda )$, the $i\pi$-quasi-periodicity condition of the function \eqref{def-bart} can be rewritten as
\begin{equation}
\frac{-a(\lambda )\,Q(\lambda -\eta )+d(\lambda )\,Q(\lambda +\eta )}{Q(\lambda )}
=\frac{a(\lambda )\,Q(\lambda +i\pi -\eta )-d(\lambda )\,Q(\lambda +i\pi +\eta )}{Q(\lambda +i\pi )},
\end{equation}
or equivalently,
\begin{equation}\label{cond-W-ad}
a(\lambda )\,W_Q(\lambda )=d(\lambda )\,W_Q(\lambda +\eta ),
\end{equation}
with $W_Q(\lambda)$ defined by \eqref{def-W}.
Note that, for $Q(\lambda)$ of the form \eqref{Q-homo-Ansatz}, $W_Q(\lambda)$ takes the form
\begin{multline}\label{W-form}
  W_Q(\lambda)=\left(\frac{i}{2}\right)^{\!\mathsf{N}_s}\left\{\,
  \prod_{j=1}^{\mathsf{N}_s}\left[\sinh\!\left(\lambda-\lambda_j-\frac{\eta}{2}\right)-\sinh\frac{\eta}{2}\right]
  \right. \\
  \left.
  +\prod_{j=1}^{\mathsf{N}_s}\left[\sinh\!\left(\lambda-\lambda_j-\frac{\eta}{2}\right)+\sinh\frac{\eta}{2}\right]
  \right\},
\end{multline}
and therefore belongs to $\hat{\mathbb{T}}_{\mathsf{N}_s}[\lambda]$; hence it can be written in the form 
\begin{equation}\label{W-d-form}
   W_Q(\lambda)=c_W\prod_{j=1}^{\mathsf{N}_s}\sinh(\lambda-\tilde\lambda_j),
\end{equation}
where the roots $\tilde\lambda_j$ and the normalization constant $c_W$ are functions of the roots $\lambda_k$ of $Q(\lambda)$.
From the explicit expression \eqref{a-d} of $a(\lambda)$ and $d(\lambda)$, it is therefore simple to verify that the condition \eqref{cond-W-ad} is satisfied if and only if the set of roots $\{\tilde\lambda_j;1\le j\le \mathsf{N}_s\}$ of $W_Q(\lambda)$ coincides with the set of shifted inhomogeneity parameters $\big\{\xi_j^{(h_j)}; 1\le j\le \mathsf{N}, 0\le h_j\le 2s_j-1\big\}$.
The fact that $c_{W}=2\, \epsilon \, (i/2)^{\mathsf{N}_s}$ for some $\epsilon \in\{-1,1\}$ follows (together with the sum-rule \eqref{asymp-cond}) from the comparison of the $\lambda\to\pm\infty$ leading asymptotic behavior of the two expressions of $W_Q(\lambda)$, namely \eqref{def-W} using \eqref{Q-homo-Ansatz} and \eqref{W-d-form} with the roots identified with the shifted inhomogeneity parameters. 
\end{proof}

\bigskip
\noindent
\textit{Proof of Theorem~\ref{th-hom}.} 
 Let us first show that  {\it \ref{cond2-hom}.} implies {\it \ref{cond1-hom}.}
 We recall that, if $Q(\lambda)$ of the form \eqref{Q-homo-Ansatz} satisfies the homogeneous functional equation \eqref{Bax-hom-eq}, then $Q(\lambda+i\pi)$ satisfies the homogeneous functional equation \eqref{hom-eq-bis}.
Particularizing \eqref{Bax-hom-eq} and \eqref{hom-eq-bis} to the points $\xi_n^{(h_n)}$, $n\in\{1,\ldots,\mathsf{N}\}$, $h_n\in\{0,1,\ldots,2s_n\}$, we obtain that
\begin{equation}
D_{\bar{\mathsf{t}}}^{(n)}
\cdot
 \begin{pmatrix}
Q(\xi_n^{(0)}) \\ 
Q(\xi_n^{(1)}) \\ 
\vdots \\ 
Q(\xi_n^{(2s_{n})})
\end{pmatrix}
=
  X_{i\pi}^{-1}\, D_{\bar{\mathsf{t}}}^{(n)}\, X_{i\pi}
\cdot 
\begin{pmatrix}
Q(\xi_n^{(0)}+i\pi) \\ 
Q(\xi_n^{(1)}+i\pi) \\ 
\vdots \\ 
Q(\xi_n^{(2s_{n})}+i\pi)
\end{pmatrix}
=
\begin{pmatrix}
0 \\ 
\vdots \\ 
\vdots \\ 
0
\end{pmatrix},
\quad
\forall n\in \{1,\ldots,\mathsf{N}\},
\end{equation}
where $X_{i\pi}\equiv X_{\alpha=i\pi}$ is the diagonal matrix defined as in \eqref{Xn}.
It means that, for any $n\in\{1,\ldots,\mathsf{N}\}$, both vectors
\begin{equation}
  \begin{pmatrix}
Q(\xi_n^{(0)}) \\ 
Q(\xi_n^{(1)}) \\ 
\vdots \\ 
Q(\xi_n^{(2s_{n})})
\end{pmatrix}
=\mathbf{Q}^{(n)}
\qquad\text{and}\qquad
X_{i\pi}
\cdot 
\begin{pmatrix}
Q(\xi_n^{(0)}+i\pi) \\ 
Q(\xi_n^{(1)}+i\pi) \\ 
\vdots \\ 
Q(\xi_n^{(2s_{n})}+i\pi)
\end{pmatrix}
=\tilde{\mathbf{Q}}^{(n)}
\end{equation}
belong to the kernel of the matrix $D_{\bar{\mathsf{t}}}^{(n)}$. Since by hypothesis at least one of them is non-zero, it means that $\bar{\mathsf{t}}(\lambda)$ satisfies the system \eqref{discrete-sys}.
 Moreover, from the form of the functional equation \eqref{Bax-hom-eq} and the fact that $\bar{\mathsf{t}}(\lambda)$ is entire, it follows that $\bar{\mathsf{t}}(\lambda)$ belongs to $\mathbb{T}_{2\mathsf{N}-2}[\frac{\lambda}{2}]$. Since $\bar{\mathsf{t}}(\lambda)$ satisfies in addition the quasi-periodicity property $\mathsf{\bar{t}}(\lambda+i\pi)=(-1)^{\mathsf{N}-1}\,\mathsf{\bar{t}}(\lambda)$, it is of the form \eqref{t-form} and therefore belongs to $\Sigma_{\bar{\mathsf{T}}}$. Hence   {\it \ref{cond2-hom}.}  implies   {\it \ref{cond1-hom}.} 
 
 \medskip
 
 Let us now show that   {\it \ref{cond1-hom}.}  implies   {\it \ref{cond2-hom}.} 
 
 Let $\bar{\mathsf{t}}(\lambda)\in\Sigma_{\bar{\mathsf{T}}}$. Let $\epsilon\in\{0,1\}$ such that $\mathsf{N}_s\in(2\mathbb{Z}+\epsilon)$. We consider the following linear map:
 \begin{equation}\label{def-mapF}
 \begin{aligned}
 \mathcal{F}\colon 
 \mathbb{T}^{(\epsilon)}[\textstyle{\frac{\lambda}{2}}]  &\,\to\, 
 \mathbb{T}^{(\epsilon)}[\textstyle{\frac{\lambda}{2}}] 
 \\
  Q(\lambda)\ \ &\,\mapsto \,
     F_Q(\lambda)\equiv
      \mathsf{\bar{t}}(\lambda)\, Q(\lambda)
      + a(\lambda)\, Q(\lambda-\eta)
      - d(\lambda)\, Q(\lambda+\eta).
 \end{aligned}
 \end{equation}
 It is easy to see, from the consideration of the leading asymptotic behavior of $F_Q(\lambda)$ when $\lambda\to \pm\infty$, that
 \begin{itemize}
 \item if $Q(\lambda)\in\mathbb{T}_{\mathsf{N}_s}[\frac{\lambda}{2}]$, then $F_Q(\lambda)\in\mathbb{T}_{\mathsf{N}_s+2\mathsf{N}-2}[\frac{\lambda}{2}]$ ;
 \item if $Q(\lambda)\in\mathbb{T}^{(\epsilon)}[\frac{\lambda}{2}]\setminus\mathbb{T}_{\mathsf{N}_s}[\frac{\lambda}{2}]$, then $F_Q(\lambda)\in\mathbb{T}^{(\epsilon)}[\frac{\lambda}{2}]\setminus\mathbb{T}_{\mathsf{N}_s+2\mathsf{N}}[\frac{\lambda}{2}]$ ;
 \item $\ker\mathcal{F}\subset\hat{\mathbb{T}}_{\mathsf{N}_s}[\frac{\lambda}{2}]\cup\{0\}$.
 \end{itemize}
 We shall moreover prove that $\dim\ker\mathcal{F}\ge 1$, which (together with Lemma~\ref{lem-uniqueness}) will prove the first assertion of {\it \ref{cond2-hom}.} 
 
 Let us first consider the subspace $\mathbb{V}_1$ of $\mathbb{T}^{(\epsilon)}[\frac{\lambda}{2}]$ defined as
 \begin{equation}
    \mathbb{V}_1=\left\{ Q(\lambda)\in \mathbb{T}_{\mathsf{N}_s}[\textstyle{\frac{\lambda}{2}}] \mid
    F_Q(\xi_n^{(h_n)})=0,\ 1\le n\le \mathsf{N},\ 0\le h_n\le 2s_n
    \right\},
 \end{equation}
 i.e. the set of all functions $Q(\lambda)\in\mathbb{T}_{\mathsf{N}_s}[\frac{\lambda}{2}]$ such that
\begin{equation}\label{sub-sys-hom}
D_{\bar{\mathsf{t}}}^{(n)}
\cdot
\mathbf{Q}^{(n)}= 0,
\qquad
\forall n\in \{1,\ldots,\mathsf{N}\}.
\end{equation}
Decomposing $Q(\lambda)$ on the basis of $\mathbb{T}_{\mathsf{N}_s}[\textstyle{\frac{\lambda}{2}}] $ with elements
\begin{equation}\label{basis-TNs}
   \prod_{\substack{ \ell=0  \\ \ell\neq k }}^{\mathsf{N}_s}
     \frac{\sinh \big(\frac{\lambda -\zeta_{\ell}}{2}\big)}{\sinh \big(\frac{\zeta_k-\zeta_\ell}{2}\big)},
     \qquad 0\le k\le \mathsf{N}_s,
\end{equation}
where $\zeta_{j_{n,h}}\equiv \xi_n^{(j)}$, $j_{n,h}=2\sum_{k=1}^{n-1}s_k+h+1$, for  $1\le n \le \mathsf{N}$, $0\le h \le 2s_n-1$, $\zeta_0$ being an arbitrary complex number (distinct from $\zeta_1,\ldots,\zeta_{\mathsf{N}_s}$ modulo $2\pi i$), we obtain that the elements of $\mathbb{V}_1$ are given by the solutions to the union of the following coupled sub-systems, for $1\le n\le \mathsf{N}$:
\begin{equation}\label{sub-sys2-hom}
  \left\{
  \begin{split}
   &Q(\xi_n^{(h)})= \mathsf{q}_{\,\bar{\mathsf{t}},h}^{(n)}\ Q(\xi_n^{(0)}),
     \qquad 1\le h\le 2s_n,\\
   & Q(\xi_n^{(2s_n)})
     =\sum_{k=0}^{\mathsf{N_s}}\ \prod_{\substack{ \ell=0  \\ \ell\neq k }}^{\mathsf{N}_s}
     \frac{\sinh \big(\frac{\xi_n^{(2s_n)} -\zeta_{\ell}}{2}\big)}{\sinh \big(\frac{\zeta_k-\zeta_\ell}{2}\big)} \ Q(\zeta_k).
  \end{split}
  \right.
\end{equation}
The first line in \eqref{sub-sys2-hom} is just a rewriting of \eqref{sub-sys-hom}, whereas the second line encodes the fact that $Q(\lambda)\in\mathbb{T}_{\mathsf{N}_s}[\frac{\lambda}{2}]$.
It is easy to see that the space of solutions to \eqref{sub-sys2-hom} is isomorphic to the space of solutions to the following homogeneous system of $\mathsf{N}$ equations in $\mathsf{N}+1$ unknowns $q_0\equiv Q(\zeta_0)$ and $q_j\equiv Q(\xi_j^{(0)})$, $1\le j\le \mathsf{N}$,
\begin{equation}\label{sys-0-2-hom}
   \sum_{j=0}^\mathsf{N} \left[\,\overline{\mathcal{N}}_{\bar{\mathsf{t}}}(\zeta_0)\right]_{i,j}\, q_j
   = 0,
   \qquad
   1\le i\le \mathsf{N},
\end{equation}
with
\begin{align}
  &\left[\,\overline{\mathcal{N}}_{\bar{\mathsf{t}}}(\zeta)\right]_{i,0}
    =-\prod_{k=1}^\mathsf{N}\prod_{\ell=0}^{2s_k-1}
      \frac{\sinh\Big(\frac{\xi_i^{(2s_i)}-\xi_k^{(\ell)}}{2}\Big)}{\sinh\Big(\frac{\zeta-\xi_k^{(\ell)}}{2}\Big)},
      \qquad 1\le i\le \mathsf{N},
      \label{mat-N-bar0}\\
  &\left[\,\overline{\mathcal{N}}_{\bar{\mathsf{t}}}(\zeta)\right]_{i,j}
    =\delta_{i,j}\  \mathsf{q}_{\,\bar{\mathsf{t}},2s_i}^{(i)} \notag\\
  &\hspace{1.8cm}
  -\sum_{h=0}^{2s_j-1}  \frac{\sinh\Big(\frac{\xi_i^{(2s_i)}-\zeta}{2}\Big)}{\sinh\Big(\frac{\xi_j^{(h)}-\zeta}{2}\Big)}\,
  \underset{(k,\ell)\not=(j,h)\ }{\prod_{k=1}^\mathsf{N}\prod_{\ell=0}^{2s_k-1}}
  \frac{\sinh\Big(\frac{\xi_i^{(2s_i)}-\xi_k^{(\ell)}}{2}\Big)}{\sinh\Big(\frac{\xi_j^{(h)}-\xi_k^{(\ell)}}{2}\Big)}\ 
  \mathsf{q}_{\bar{\mathsf{t}},h}^{(j)},
  \quad 1\le i,j\le \mathsf{N}.
      \label{mat-N-bar1}
\end{align}
Hence, the dimension of the space  $\mathbb{V}_1$ is given by the dimension of the nullspace of the $\mathsf{N}\times(\mathsf{N}+1)$ matrix $\overline{\mathcal{N}}_{\bar{\mathsf{t}}}(\zeta_0)$, i.e.
\begin{equation}
   \dim\mathbb{V}_1=\dim\ker\overline{\mathcal{N}}_{\bar{\mathsf{t}}}(\zeta_0)
   =\mathsf{N}+1-\mathrm{rk}\,\overline{\mathcal{N}}_{\bar{\mathsf{t}}}(\zeta_0) \ge 1, 
\end{equation}
where $\mathrm{rk}\,\overline{\mathcal{N}}_{\bar{\mathsf{t}}}(\zeta_0)\le \mathsf{N}$ denotes the rank of the matrix $\overline{\mathcal{N}}_{\bar{\mathsf{t}}}(\zeta_0)$.
Note at this point that it is always possible to choose $\zeta_0$ in such a way that the rank of the matrices $\overline{\mathcal{N}}_{\bar{\mathsf{t}}}(\zeta_0)$ and $\mathcal{N}_{\bar{\mathsf{t}}}(\zeta_0)$, where $\mathcal{N}_{\bar{\mathsf{t}}}(\zeta_0)$ is the $\mathsf{N}\times\mathsf{\mathsf{N}}$ matrix consisting of the last $\mathsf{N}$ columns of  $\overline{\mathcal{N}}_{\bar{\mathsf{t}}}(\zeta_0)$ (i.e. of the elements \eqref{mat-N-bar1} for $1\le i,j\le \mathsf{N}$), are equal\footnote{Indeed, if $Q(\lambda)$ is a non-zero solution to the above system (which exists since $\dim\mathbb{V}_1\ge 1$), it is enough to choose $\zeta_0$ such that $Q(\zeta_0)\not=0$.}, so that we have also
\begin{equation}
   \dim\mathbb{V}_1=\mathsf{N}+1-\mathrm{rk}\, {\mathcal{N}}_{\bar{\mathsf{t}}}(\zeta_0).
\end{equation}

Let us now consider the restriction $\mathcal{F}_{|_{\mathbb{V}_1}}$ of $\mathcal{F}$ to $\mathbb{V}_1$. Since $\ker\mathcal{F}\subset\mathbb{V}_1$, we have that $\ker\mathcal{F}=\ker\mathcal{F}_{|_{\mathbb{V}_1}}$. It remains to compute the dimension of the image $\mathrm{Im}\,\mathcal{F}_{|_{\mathbb{V}_1}}$ of $\mathcal{F}_{|_{\mathbb{V}_1}}$.

Let $Q(\lambda)\in\mathbb{V}_1$. Then $F_Q(\lambda)\in\mathbb{T}_{\mathsf{N}_s+2\mathsf{N}-2}[\frac{\lambda}{2}]$ and vanishes at the points $\xi_n^{(h_n)}$, $n\in\{1,\ldots,\mathsf{N}\}$, $h_n\in\{0,1,\ldots,2s_n\}$, i.e. it is of the form
\begin{equation}\label{rep-F_Q}
  F_Q(\lambda)=g(\lambda)\, \prod_{n=1}^{\mathsf{N}}\prod_{h_n=0}^{2s_n}\!\sinh\Big(\frac{\lambda-\xi_n^{(h_n)}}{2}\Big),
\end{equation}
with $g(\lambda)\in\mathbb{T}_{\mathsf{N}-2}[\frac{\lambda}{2}]$.
Moreover, its values at the points $\xi_n^{(h_n)}+i\pi$, $n\in\{1,\ldots,\mathsf{N}\}$, $h_n\in\{0,1,\ldots,2s_n\}$, are given in terms of the values of $Q(\lambda)$ in these points by
%
\begin{equation}\label{sub-sys-F}
  X_{i\pi}^{-1}\, D_{\bar{\mathsf{t}}}^{(n)}\, X_{i\pi}
\cdot
\begin{pmatrix}
Q(\xi_n^{(0)}+i\pi) \\ 
Q(\xi_n^{(1)}+i\pi) \\ 
\vdots \\ 
Q(\xi_n^{(2s_{n})}+i\pi)
\end{pmatrix}
=
(-1)^{\mathsf{N}-1}
\begin{pmatrix}
F_Q(\xi_n^{(0)}+i\pi) \\ 
F_Q(\xi_n^{(1)}+i\pi) \\ 
\vdots \\ 
F_Q(\xi_n^{(2s_n)}+i\pi)
\end{pmatrix},
\end{equation}
where we have used the $i\pi$-quasi-periodicity of $\bar{\mathsf{t}}(\lambda)$, $a(\lambda)$ and $d(\lambda)$.
Equation \eqref{sub-sys-F} means that, for each  $n\in \{1,\ldots,\mathsf{N}\}$, the vector
\begin{equation}
  (-1)^{\mathsf{N}-1}
  X_{i\pi}
\cdot
\begin{pmatrix}
F_Q(\xi_n^{(0)}+i\pi) \\ 
F_Q(\xi_n^{(1)}+i\pi) \\ 
\vdots \\ 
F_Q(\xi_n^{(2s_n)}+i\pi)
\end{pmatrix}
\end{equation}
belongs to the image of the matrix $D_{\bar{\mathsf{t}}}^{(n)}$, i.e. is orthogonal to the row vector ${}^t \mathbf{p}_{\,\bar{\mathsf{t}}}^{(n)}$ \eqref{left-vector}:
\begin{equation}\label{cond-F_Q}
   {}^t \mathbf{p}_{\,\bar{\mathsf{t}}}^{(n)} \, X_{i\pi}\cdot
   \begin{pmatrix}
F_Q(\xi_n^{(0)}+i\pi) \\ 
F_Q(\xi_n^{(1)}+i\pi) \\ 
\vdots \\ 
F_Q(\xi_n^{(2s_n)}+i\pi)
\end{pmatrix} =0,
\qquad \forall n\in \{1,\ldots,\mathsf{N}\}.
\end{equation}
Using the expression \eqref{rep-F_Q}, the condition \eqref{cond-F_Q} can be rewritten as
\begin{equation}\label{cond-g}
  \sum_{h=0}^{2s_n}  \mathsf{z}_{\bar{\mathsf{t}},h}^{(n)} \ g(\xi_n^{(h)}+i\pi)=0,
  \qquad \forall n\in \{1,\ldots,\mathsf{N}\},
\end{equation}
where
\begin{align*}
  \mathsf{z}_{\bar{\mathsf{t}},h}^{(n)}
  &=\prod_{j=1}^{\mathsf{N}}\prod_{h_j=0}^{2s_j}\!\sinh\Big(\frac{\xi_n^{(h)}+i\pi-\xi_j^{(h_j)}}{2}\Big)\
  \prod_{\ell=0}^{h-1}\frac{a(\xi_n^{(\ell)})}{d(\xi_n^{(\ell+1)})}\
  \mathsf{q}_{\bar{\mathsf{t}},h}^{(n)},
  \\
  &=\prod_{k=1}^\mathsf{N}
  \frac{\prod_{\ell=1}^{2s_k}\sinh\Big(\frac{\xi_n^{(0)}-\xi_k^{(0)}+\ell\eta}{2}\Big)
           \prod_{\ell=0}^{2s_k}\sinh\Big(\frac{\xi_n^{(0)}-\xi_k^{(0)}+\ell\eta+i\pi}{2}\Big)}
          {\prod_{\substack{\ell=0\\ (k,\ell)\not= (n,h)}}^{2s_k}\sinh\Big(\frac{\xi_n^{(h)}-\xi_k^{(\ell)}}{2}\Big)}
   \prod_{\substack{k=1 \\ k\not= n}}^\mathsf{N} \sinh\Big(\frac{\xi_n^{(0)}-\xi_k^{(0)}}{2}\Big)\
  \mathsf{q}_{\bar{\mathsf{t}},h}^{(n)}.
\end{align*}
Note that this system can be simplified as
\begin{equation}\label{cond-g-bis}
  \sum_{h=0}^{2s_n}  \tilde{\mathsf{z}}_{\bar{\mathsf{t}},h}^{(n)} \ g(\xi_n^{(h)}+i\pi)=0,
  \qquad \forall n\in \{1,\ldots,\mathsf{N}\},
\end{equation}
where
\begin{equation}\label{z_t}
   \tilde{\mathsf{z}}_{\bar{\mathsf{t}},h}^{(n)}
   =\left[ \underset{(k,\ell)\not=(n,h)\ }{\prod_{k=1}^\mathsf{N}\prod_{\ell=0}^{2s_k}} \!\!\sinh\Big(\frac{\xi_n^{(h)}-\xi_k^{(\ell)}}{2}\Big)\right]^{-1}  \mathsf{q}_{\bar{\mathsf{t}},h}^{(n)}.
\end{equation}
Hence, $\mathrm{Im}\,\mathcal{F}_{|_{\mathbb{V}_1}}$ is a subspace of the space $\mathbb{V}_2$ consisting of all functions which can be written in the form \eqref{rep-F_Q} in terms of $g(\lambda)\in\mathbb{T}_{\mathsf{N}-2}[\frac{\lambda}{2}]$ satisfying the condition \eqref{cond-g-bis}.


For $g(\lambda)\in\mathbb{T}_{\mathsf{N}-2}[\frac{\lambda}{2}]$ and $\zeta\in\mathbb{C}$ arbitrary, let us consider the function $g_\zeta(\lambda)\in\mathbb{T}_{\mathsf{N}-1}[\frac{\lambda}{2}]$ defined as
\begin{equation}
   g_\zeta(\lambda)\equiv \sinh\!\Big(\frac{\lambda-\zeta-i\pi}{2}\Big)\ g(\lambda).
\end{equation}
It can be expressed as in \eqref{interpol-P} in terms of its values in $\mathsf{N}$ different points modulo $2\pi i$, say the points $\xi_j^{(2s_j)}+i\pi$, $j\in\{1,\ldots,\mathsf{N}\}$, so that the function $g(\lambda)$ admits the following representation:
\begin{equation}\label{rep-g}
  g(\lambda+i\pi)=\sum_{j=1}^\mathsf{N}
  \frac{\sinh\Big(\frac{\xi_j^{(2s_j)}-\zeta}{2}\Big)}{\sinh\big(\frac{\lambda-\zeta}{2}\big)}
  \prod_{\substack{k=1 \\ k\not= j}}^\mathsf{N}
  \frac{\sinh\Big(\frac{\lambda-\xi_k^{(2s_k)} }{2}\Big)}{\sinh\Big(\frac {\xi_j^{(2s_j)}-\xi_k^{(2s_k)}}{2}\Big)}\
  g\big(\xi_j^{(2s_j)}+i\pi\big).
\end{equation}
Using the representation \eqref{rep-g}, one can rewrite \eqref{cond-g-bis}-\eqref{z_t} in the form of a homogeneous system of $\mathsf{N}$ linear equations for the $\mathsf{N}$ values $g\big(\xi_j^{(2s_j)}+i\pi\big)$, $1\le j\le \mathsf{N}$,
\begin{equation}\label{sys-inh-hom}
  \sum_{j=1}^\mathsf{N} \big[\widetilde{\mathcal{N}}_{\bar{\mathsf{t}}}(\zeta)\big]_{i,j}\ 
  g\big(\xi_j^{(2s_j)}+i\pi\big)=0,
  \qquad 1\le i\le \mathsf{N},
\end{equation}
with matrix
\begin{align}
  \big[\widetilde{\mathcal{N}}_{\bar{\mathsf{t}}}(\zeta)\big]_{i,j}
  &=\delta_{i,j}\   \tilde{\mathsf{z}}_{\bar{\mathsf{t}},2s_j}^{(j)}
     +\sum_{h=0}^{2s_i-1} \frac{\sinh\Big(\frac{\xi_j^{(2s_j)}-\zeta}{2}\Big)}{\sinh\Big(\frac{\xi_i^{(h)}-\zeta}{2}\Big)}\,
  \prod_{\substack{k=1 \\ k\not=j}}^{\mathsf{N}}
  \frac{\sinh\Big(\frac{\xi_i^{(h)}-\xi_k^{(2s_k)}}{2}\big)}{\sinh\Big(\frac{\xi_j^{(2s_j)}-\xi_k^{(2s_k)}}{2}\Big)}\ \tilde{\mathsf{z}}_{\bar{\mathsf{t}},h}^{(i)}
  \notag\\
  &=\left[ \underset{(k,\ell)\not=(j,2s_j)\ }{\prod_{k=1}^\mathsf{N}\prod_{\ell=0}^{2s_k}} \!\!\sinh\Big(\frac{\xi_j^{(2s_j)}-\xi_k^{(\ell)}}{2}\Big)\right]^{-1}
       \left\{\delta_{i,j}\  \mathsf{q}_{\bar{\mathsf{t}},2s_j}^{(j)}
       \vphantom{\underset{(k,\ell)\not=(j,2s_j)\ }{\prod_{k=1}^\mathsf{N}\prod_{\ell=0}^{2s_k}}}
       \right.
       \notag\\
  &\hspace{2cm}\left.
       -\sum_{h=0}^{2s_i-1}
       \frac{\sinh\Big(\frac{\xi_j^{(2s_j)}-\zeta}{2}\Big)}{\sinh\Big(\frac{\xi_i^{(h)}-\zeta}{2}\Big)}\,
  \underset{(k,\ell)\not=(i,h)\ }{\prod_{k=1}^\mathsf{N}\prod_{\ell=0}^{2s_k-1}}
  \frac{\sinh\Big(\frac{\xi_j^{(2s_j)}-\xi_k^{(\ell)}}{2}\Big)}{\sinh\Big(\frac{\xi_i^{(h)}-\xi_k^{(\ell)}}{2}\Big)}\ 
  \mathsf{q}_{\bar{\mathsf{t}},h}^{(i)}
       \right\}.
  \label{mat-sys-g}
\end{align}
Hence, the dimension of the space $\mathbb{V}_2$, which corresponds to the dimension of the kernel of the $\mathsf{N}\times\mathsf{N}$ matrix $\widetilde{\mathcal{N}}_{\bar{\mathsf{t}}}(\zeta)$, is
\begin{equation}
   \dim\mathbb{V}_2=\dim\ker\widetilde{\mathcal{N}}_{\bar{\mathsf{t}}}(\zeta)
   =\mathsf{N}-\mathrm{rk}\,\widetilde{\mathcal{N}}_{\bar{\mathsf{t}}}(\zeta), 
\end{equation}
where $\mathrm{rk}\,\widetilde{\mathcal{N}}_{\bar{\mathsf{t}}}(\zeta)$ denotes the rank of the matrix $\widetilde{\mathcal{N}}_{\bar{\mathsf{t}}}(\zeta)$ \eqref{mat-sys-g}.
Note at this point that this matrix can be obtained from the transpose of the matrix $\mathcal{N}_{\bar{\mathsf{t}}}(\zeta)$ with elements \eqref{mat-N-bar1} by a similarity transformation, which means that $\widetilde{\mathcal{N}}_{\bar{\mathsf{t}}}(\zeta)$ and ${\mathcal{N}}_{\bar{\mathsf{t}}}(\zeta)$ have the same rank. This implies that
\begin{equation}
  \dim\mathrm{Im}\,\mathcal{F}_{|_{\mathbb{V}_1}}\le \dim\mathbb{V}_2=\dim\mathbb{V}_1-1,
\end{equation}
so that $\dim\ker\mathcal{F}=\dim\ker\mathcal{F}_{|_{\mathbb{V}_1}}\ge 1$.
In other words, it means that there exists, up to normalization, at least one non-zero element $Q(\lambda)$ of $\mathbb{T}^{(\epsilon)}[\frac{\lambda}{2}]$ (which belongs to $\hat{\mathbb{T}}_{\mathsf{N}_s}[\frac{\lambda}{2}]$ due to the previous remark on the leading asymptotic behavior of $F_Q(\lambda)$) such that $F_Q(\lambda)$ vanishes identically, i.e. that there exists at least one solution $Q(\lambda)$ of the form \eqref{Q-homo-Ansatz} to the homogeneous functional equation \eqref{Bax-hom-eq}.
The uniqueness of this solution follows from Lemma~\ref{lem-uniqueness}.

Let us finally suppose that, for this solution, there exists $n\in\{1,\ldots,\mathsf{N}\}$ such that $Q(\xi_n^{(0)})=Q(\xi_n^{(0)}+i\pi)=0$. This implies that, $\forall h\in\{0,\ldots,2s_n\}$, $Q(\xi_n^{(h)})=Q(\xi_n^{(h)}+i\pi)=0$, i.e. that $Q(\lambda)$ is of the form
\begin{equation}
   Q(\lambda)=\prod_{h=0}^{2s_n}\sinh(\lambda-\xi_n^{(h)})\ Q_n(\lambda),
\end{equation}
where ${Q}_n(\lambda)\in\hat{\mathbb{T}}_{\mathsf{N}_n}[\frac{\lambda}{2}]$ with $\mathsf{N}_n\equiv\mathsf{N}_s-2(2s_n+1)$. 
Then
\begin{equation}\label{W_Q-n}
   W_Q(\lambda)= (-1)^{2s_n+1}\prod_{h=0}^{2s_n}\left[ \sinh(\lambda-\xi_n^{(h)})\, \sinh(\lambda-\xi_n^{(h)}-\eta)\right]\, W_{Q_n}(\lambda),
\end{equation}
where $W_Q(\lambda)\in\mathbb{T}_{\mathsf{N}_s}[\lambda]$ and $W_{Q_n}(\lambda)\in\mathbb{T}_{\mathsf{N}_n}[\lambda]$ are defined from $Q(\lambda)$ and $Q_n(\lambda)$ as in \eqref{def-W}.
It is then enough to notice than the form \eqref{W_Q-n} of $W_Q(\lambda)$ is incompatible with the relation \eqref{W-relation} that should be satisfied for $\bar{\mathsf{t}}(\lambda)\in\Sigma_{\bar{\mathsf{T}}}$ (see Lemma~\ref{lem-WQ}): we have $W_Q(\xi_n^{(2s_n)})=W_Q(\xi_n^{(0)}+\eta)=0$, whereas the r.h.s. of \eqref{W-relation} is obviously non-zero for $\lambda=\xi_n^{(2s_n)}$ or $\lambda=\xi_n^{(0)}+\eta$, which ends the proof of    {\it \ref{cond2-hom}.} 

\medskip

Let us finally show the equivalence between  {\it \ref{cond2-hom}.}  and  {\it \ref{cond3-hom}.}

Let $\bar{\mathsf{t}}(\lambda)$ satisfy the condition {\it \ref{cond2-hom}.} and let $Q(\lambda)$ be the unique solution of the form \eqref{Q-homo-Ansatz} of the homogeneous equation \eqref{Bax-hom-eq} associated to $\bar{\mathsf{t}}(\lambda)$.
Then Lemma~\ref{lem-WQ} implies that $Q(\lambda)$ satisfies the Wronskian identity \eqref{W-relation} for some $\epsilon\in\{-1,1\}$. One can then rewrite the functions $a(\lambda)$ and $d(\lambda)$ in terms of $W_Q(\lambda)$ as
\begin{equation}\label{ad-WQ}
a(\lambda)=W_Q(\lambda+\eta)/ w_\epsilon(\lambda ),\qquad 
d(\lambda)=W_Q(\lambda)/ w_\epsilon(\lambda ),
\end{equation}
%
where $w_\epsilon(\lambda)$ is defined by \eqref{w-eps},
and re-inject these expressions into the homogeneous functional equation \eqref{Bax-hom-eq} satisfied by 
$\bar{\mathsf{t}}(\lambda)$ and $Q(\lambda)$. It immediately follows that $\bar{\mathsf{t}}(\lambda)$ is of the form \eqref{expr-t-Q}.

Conversely, let  $Q(\lambda)$ be a solution of the form \eqref{Q-homo-Ansatz} to the Wronskian equation \eqref{W-relation} for some $\epsilon\in\{-1,1\}$, and let $\bar{\mathsf{t}}(\lambda)$ be expressed in terms of $Q(\lambda)$ and $\epsilon$ as in \eqref{expr-t-Q}.
Then, multiplying both members of \eqref{expr-t-Q} by $Q(\lambda)$, and adding and subtracting the expression
\begin{equation}
  Q(\lambda+i\pi)\,Q(\lambda-\eta)\,Q(\lambda+\eta)
\end{equation}
to the r.h.s. so as to reconstruct the functions $a(\lambda)$ and $d(\lambda)$ from their representations \eqref{ad-WQ}, we obtain that $\bar{\mathsf{t}}(\lambda)$ and $Q(\lambda)$ satisfy the homogeneous functional equation \eqref{Bax-hom-eq}. The fact $\bar{\mathsf{t}}(\lambda)$ satisfies the quasi-periodicity property $\mathsf{\bar{t}}(\lambda+i\pi)=(-1)^{\mathsf{N}-1}\,\mathsf{\bar{t}}(\lambda)$ is a trivial consequence of \eqref{expr-t-Q}.

Finally, it follows from these considerations that the uniqueness of the couple $(\epsilon,Q)$ in the statement of {\it \ref{cond3-hom}.} implies the uniqueness of the solution $Q(\lambda)$ in {\it \ref{cond2-hom}.}, and that conversely, the uniqueness of $Q(\lambda)$ in {\it \ref{cond2-hom}.} implies the uniqueness of  $(\epsilon,Q)$ in {\it \ref{cond3-hom}.}

This ends the proof of the theorem.
\qed
\bigskip

We would like to conclude this section by proving what has been announced  in Remark~\ref{rem-prop-vec}, namely the fact that  {\it \ref{cond3-hom}.} of Theorem~\ref{th-hom} admits a simple geometrical reformulation in terms of the proportionality condition of the two vectors \eqref{def-Q-tildeQ} for each $n\in\{1,\ldots,\mathsf{N}\}$:

\begin{proposition}\label{prop-WQ}
Let $Q(\lambda)$ be a function of the form \eqref{Q-homo-Ansatz} such that, for each $n\in\{1,\ldots,\mathsf{N}\}$, the two vectors $\mathbf{Q}^{(n)}$ and $\tilde{\mathbf{Q}}^{(n)}$ \eqref{def-Q-tildeQ} are proportional and not both equal to zero.
Then $Q(\lambda)$ satisfies the Wronskian identity \eqref{W-relation} for some $\epsilon\in\{-1,1\}$, and the function $\bar{\mathsf{t}}(\lambda)$ defined by \eqref{expr-t-Q} belongs to $\Sigma_{\bar{\mathsf{T}}}$.

Conversely, let $\bar{\mathsf{t}}(\lambda)\in\Sigma_{\bar{\mathsf{T}}}$. 
Then $\bar{\mathsf{t}}(\lambda)$ can be expressed as in \eqref{expr-t-Q} in terms of some $\epsilon\in\{-1,1\}$ and of some function $Q(\lambda)$ of the form \eqref{Q-homo-Ansatz} such that, for each $n\in\{1,\ldots,\mathsf{N}\}$, the two vectors $\mathbf{Q}^{(n)}$ and $\tilde{\mathbf{Q}}^{(n)}$ \eqref{def-Q-tildeQ} are proportional and not both equal to zero.
\end{proposition}

\begin{proof}
The linear dependence of the two vectors $\mathbf{Q}^{(n)}$ and $\tilde{\mathbf{Q}}^{(n)}$ \eqref{def-Q-tildeQ} implies the vanishing of each of the following determinants,
\begin{equation}
  \begin{vmatrix}
    Q(\xi_n^{(h_n)}) & Q(\xi_n^{(h_n)}+i\pi)\\
    Q(\xi_n^{(h_n)}-\eta) & -Q(\xi_n^{(h_n)}-\eta+i\pi)
  \end{vmatrix},
  \qquad
  1\le n\le \mathsf{N},   \quad  0\le h_n\le 2s_n-1,
\end{equation}
and hence the vanishing of the function $W_Q(\lambda)$ \eqref{def-W} at the $\mathsf{N}_s$ points $\lambda=\xi_n^{(h_n)}$, $1\le n\le \mathsf{N}$, $0\le h_n\le 2s_n-1$.
We can moreover show, like in the proof of Lemma~\ref{lem-WQ} that, for $Q(\lambda)$ of the form \eqref{Q-homo-Ansatz}, this function $W_Q(\lambda)$ is of the form \eqref{W-d-form}. Hence the set of roots of $W_Q(\lambda)$ coincides exactly with the set of shifted inhomogeneity parameters $\big\{\xi_n^{(h_n)}; 1\le n\le \mathsf{N}, 0\le h_n\le 2s_n-1\big\}$. The normalization coefficient $c_W$ is then fixed as in the proof of Lemma~\ref{lem-WQ}, and therefore  $Q(\lambda)$ satisfies the Wronskian equation \eqref{W-relation} for some $\epsilon\in\{-1,1\}$.
Finally, the proportionality of the two vectors \eqref{def-Q-tildeQ} for each $n\in\{1,\ldots,\mathsf{N}\}$ implies that the numerator of the r.h.s. of \eqref{expr-t-Q} vanishes at the points $\xi_n^{(h_n)}$, $n\in\{1,\ldots,\mathsf{N}\}$, $h_n\in\{1,\ldots,2s_n-1\}$, so that the function $\bar{\mathsf{t}}(\lambda)$ defined by \eqref{expr-t-Q} is entire.
From the proof of Theorem~\ref{th-hom}, it follows that $\bar{\mathsf{t}}(\lambda)$ and $Q(\lambda)$ satisfy the homogeneous functional equation \eqref{Bax-hom-eq}, and that $\bar{\mathsf{t}}(\lambda)\in\Sigma_{\bar{\mathsf{T}}}$. 

Conversely, if $\bar{\mathsf{t}}(\lambda)\in\Sigma_{\bar{\mathsf{T}}}$, then $\bar{\mathsf{t}}(\lambda)$ can be expressed as in \eqref{expr-t-Q} in terms of some function $Q(\lambda)$ of the form \eqref{Q-homo-Ansatz} satisfying the Wronskian identity \eqref{W-relation}.
It follows from the proof of Theorem~\ref{th-hom} that $\mathsf{\bar{t}}(\lambda)$ and $Q(\lambda)$ satisfy the functional equation \eqref{Bax-hom-eq}. The proof of the proportionality of the two vectors $\mathbf{Q}^{(n)}$ and $\widetilde{\mathbf{Q}}^{(n)}$ is then derived as discussed after their definition \eqref{def-Q-tildeQ}.
 \end{proof}

\section{Conclusion}

In this paper, we have studied the inhomogeneous antiperiodic integrable XXZ chain with arbitrary spins in the SOV framework.
We have obtained a complete characterization of the  ($\kappa$-twisted) antiperiodic transfer matrix spectrum and eigenvectors in terms of solutions of systems of discrete finite-difference equations (Theorem~\ref{th-eigen-t}). We have then reformulated this SOV spectral description in terms of particular solutions
of finite-difference functional equations of Baxter's type, hence showing the completeness of the corresponding systems of Bethe-type equations.

We have in fact considered two different reformulations of the SOV characterization of the spectrum of Theorem~\ref{th-eigen-t}.
The first one is in terms of  the solutions of a one-parameter family of {\em inhomogeneous} finite-difference functional equations for trigonometric polynomials $Q(\la)$ of degree $\mathsf{N}_s=2\sum_{n=1}^\mathsf{N}s_n$, where $s_n$ denotes the value of the spin at site $n$ (Theorem~\ref{th-inh-eq}), in the spirit of what has been recently proposed in the ``off-diagonal Bethe ansatz" framework \cite{CaoYSW13a, ZhaLCYSW14}. Such an equivalence is very natural in this context, since it amounts to solve a linear system of $\mathsf{N}$ equations in $\mathsf{N}+1$ unknowns, for which it is easy to see that there exists, under certain conditions, a unique solution. Note that this clear correspondence had already been explicitly enlightened in the case of the open spin-1/2 chain with general integrable boundaries in \cite{KitMN14}. Since the SOV approach leads to a complete description of the transfer matrix spectrum, the characterization in terms of the solutions of this corresponding second order finite-difference inhomogeneous functional equation, and therefore in terms of the solutions of the associated system of (generalized) Bethe equations, is also complete.

The second, and probably more useful, one is in terms of the solutions of the usual {\em homogeneous} Baxter's finite-difference functional equation for trigonometric polynomials of the same degree $\mathsf{N}_s$ as previously, but with double period (Theorem~\ref{th-hom}). 
These solutions have to be compared to the eigenvalues of the $Q$-operator which has been constructed in  \cite{BatBOY95,YunB95}  by means of the ``pair-propagation through a vertex" property.
We nevertheless recall that the construction of   \cite{BatBOY95,YunB95} has been performed for the homogeneous chain, a case for which Theorem~\ref{th-eigen-t} does not hold, and that it relies on a conjecture, so that we have not used this construction to show the completeness of the aforementioned  class of $Q$-solutions to \eqref{Bax-hom-eq} and of the associated system of homogeneous Bethe equations.
We have instead used simple independent arguments, not relying on the algebraic construction of the $Q$-operator, but on the study (as for the first reformulation we proposed) of some linear systems of equations. We notably expect that these arguments
could  be generalized to other situations for which the corresponding $Q$-operator is presently not known. 
It is worth mentioning at this point that the proof of the full completeness of the solutions to the homogeneous functional equation \eqref{Bax-hom-eq} is not much more complicated that for its inhomogeneous counterpart \eqref{eq-inh}, whereas the statement of the result (Theorem~\ref{th-hom} versus Theorem~\ref{th-inh-eq}) is instead considerably simpler.

It is finally important to point out that our statements of completeness are proven for inhomogeneous chains only, i.e. when the conditions $(\ref{cond-inh})$ on the inhomogeneity parameters are satisfied. This requirement is essential as it is at the basis of two properties which are fundamental for our proof of completeness: the existence of the SOV basis and the simplicity of the transfer matrix spectrum. In the homogeneous limit, the first one  does no longer hold whereas the second one is not proven in our approach\footnote{In fact, it is known in the literature \cite{PasS90,SokAC06,GunRS07,KorW07,Kor08} that under special boundary conditions and special values of the anisotropy parameter, the XXZ transfer matrix has degeneracy and can even present Jordan blocks.}.  
Of course, the $T$-$Q$ functional equations and their solutions admit a well defined limit for the homogeneous chain (this is precisely the reason of our study) and one can wonder whether they can still be used to construct the complete spectrum (eigenvalues and eigenstates)  in this limit.
To answer this question, one could in principle try to rewrite 
the SOV characterization of the transfer matrix eigenstates in such a way that their homogeneous limit is well defined.
Note however that, for our purpose of computing physical quantities such as the correlation functions, it is not very important to exhibit explicitly a complete basis of eigenstates for the homogeneous model:  the whole reasoning can indeed be performed for the inhomogeneous model, and the homogenous limit can be taken at the end of the computation only, on the scalar quantities (form factors, correlation functions) that should be continuous in this limit.

\section*{Acknowledgements}
We would like to thank J.M. Maillet, N. Kitanine and V. Pasquier for interesting discussions on early stage of this work.
G.N. and V.T. are supported by CNRS.
We also acknowledge  the support from the ANR grant DIADEMS 10 BLAN 012004.
G.N. would like to acknowledge hospitality from International Institute of Physics of Natal where part of this work has been initiated.


\appendix

\section{Trigonometric polynomials: definitions and notations}\label{app-trig-pol}

In this paper, we adopt the following terminology.

We say that a function $P(\lambda)$ is a trigonometric polynomial in $\lambda$ if it is a Laurent polynomial in $e^\lambda$. We denote by $\mathbb{T}[\lambda]\equiv \mathbb{C}[e^\lambda,e^{-\lambda}]$ the ring of all such trigonometric polynomials.

We shall more particularly focus on the class of trigonometric polynomials which are invariant (respectively negated) under a shift of $i\pi$, that we denote by
\begin{equation}\label{Teps}
  \mathbb{T}^{(\epsilon)}[\lambda]=\left\{P(\lambda)\in\mathbb{C}[e^\lambda,e^{-\lambda}] ,\ P(\lambda+i\pi)=(-1)^\epsilon P(\lambda)\right\},
  \qquad
  \epsilon\in\{0,1\}.
\end{equation}
$\mathbb{T}^{(\epsilon)}[\lambda]$ is an infinite-dimensional vector space. It consists in all functions of the form
\begin{equation}\label{def-P-M1-M2}
   P(\lambda)=e^{-\mathsf{M}_1\lambda}\sum_{j=0}^{\mathsf{M}_2} c_j\, e^{2j\lambda},
\end{equation}
for all values of $(\mathsf{M}_1,\mathsf{M}_2)\in(2\mathbb{Z}+\epsilon)\times\mathbb{N}$ and of $c_0,c_1,\ldots,c_{\mathsf{M}_2}\in\mathbb{C}$.

For any fixed pair $(\mathsf{M}_1,\mathsf{M}_2)\in\mathbb{Z}\times\mathbb{N}$, we let $\mathbb{T}_{\mathsf{M}_1,\mathsf{M}_2}[\lambda]$ denote the set of all functions of $\lambda$ of the form \eqref{def-P-M1-M2}.
If $\mathsf{M}_1\in(2\mathbb{Z}+\epsilon)$, $\mathbb{T}_{\mathsf{M}_1,\mathsf{M}_2}[\lambda]$ is a $(\mathsf{M}_2+1)$-dimensional linear subspace of $\mathbb{T}^{(\epsilon)}[\lambda]$.
When $\mathsf{M}_2=\mathsf{M}_1=\mathsf{M}$, we shall simplify the notation and set ${\mathbb{T}}_{\mathsf{M}}[\lambda]\equiv {\mathbb{T}}_{\mathsf{M},\mathsf{M}}[\lambda]$.
We shall say that an element $P(\lambda)$ of ${\mathbb{T}}_{\mathsf{M}}[\lambda]$ has degree (at most) $\mathsf{M}$.

We shall also be interested in the subset $\hat{\mathbb{T}}_{\mathsf{M}_1,\mathsf{M}_2}[\lambda]$ of ${\mathbb{T}}_{\mathsf{M}_1,\mathsf{M}_2}[\lambda]$ containing all trigonometric polynomials of the form \eqref{def-P-M1-M2} such that $c_0,c_{\mathsf{M}_2}\not=0$.
It is easy to see that the elements of $\hat{\mathbb{T}}_{\mathsf{M}_1,\mathsf{M}_2}[\lambda]$ correspond to the functions of the form
\begin{equation}\label{P-trig}
   P(\lambda)=c_P\, e^{(\mathsf{M}_2-\mathsf{M}_1)\lambda} \prod_{j=0}^{\mathsf{M}_2}\sinh(\lambda-\lambda_j),
\end{equation}
for some complex numbers $c_P$ (the normalization constant) and $\lambda_1,\ldots,\lambda_{\mathsf{M}_2}$ (the roots), with $c_P\not=0$.
In the particular case $\mathsf{M}_2=\mathsf{M}_1=\mathsf{M}$, we shall say that $P(\lambda)$ has degree exactly $\mathsf{M}$ and denote $\hat{\mathbb{T}}_{\mathsf{M}}[\lambda]\equiv\hat{\mathbb{T}}_{\mathsf{M},\mathsf{M}}[\lambda]$.

The properties of trigonometric polynomials can easily be deduced from the properties of usual polynomials.
For instance, it is easy to see that any element $P(\lambda)$ of $\mathbb{T}_{\mathsf{M}_1,\mathsf{M}_2}[\lambda]$ is completely determined by its values $P(\zeta_j)$, $j\in\{0,1,\ldots,\mathsf{M}_2\}$, in $\mathsf{M}_2+1$ arbitrary points $\zeta_0,\zeta_1,\ldots,\zeta_{\mathsf{M}_2}$   (which are pairwise distinct modulo $i\pi$) and can be written as
\begin{equation}
P(\lambda ) 
=\sum_{j=0}^{\mathsf{M}_2}
     e^{(\mathsf{M}_2-\mathsf{M}_1)(\lambda-\zeta_j)}
     \prod_{\substack{ k=0  \\ k\neq j }}^{\mathsf{M}_2}
     \frac{\sinh (\lambda -\zeta_k)}{\sinh (\zeta_j-\zeta_k)} \, P(\zeta_j).
     \label{interpol-P} 
\end{equation}
In particular, an element  $P(\lambda)$ of $\mathbb{T}_{\mathsf{M}_1,\mathsf{M}_2}[\lambda]$ which vanishes at $\mathsf{M}_2+1$ points (which are pairwise distinct modulo $i\pi$) is identically zero.
Conversely, any function of the type
\begin{equation}
F(\lambda ) 
=e^{\mathsf{m}\lambda}
  \sum_{j=0}^{\mathsf{M}_2}
     y_j
     \prod_{\substack{ k=0  \\ k\neq j }}^{\mathsf{M}_2}
     \frac{\sinh (\lambda -\zeta_k)}{\sinh (\zeta_j-\zeta_k)} ,
     \label{interpol-F} 
\end{equation}
for $\mathsf{m}\in\mathbb{Z}$ and $y_0,y_1,\ldots,y_{\mathsf{M}_2}\in\mathbb{C}$, is an element of $\mathbb{T}_{\mathsf{M}_2-\mathsf{m},\mathsf{M}_2}[\lambda]$.


\bibliographystyle{amsplain}
\bibliography{../biblio}

\end{document}